\documentclass[11pt]{article}
\usepackage[utf8]{inputenc}
\usepackage[left=1in, top=1in, right=1in, bottom=1in]{geometry}
\usepackage[T1]{fontenc}
\usepackage{libertine}

\usepackage{microtype}
\usepackage{graphicx}
\usepackage{subfigure}
\usepackage{booktabs} 
\usepackage{amsmath, amsthm, amssymb, bm}
\usepackage{mathtools}
\usepackage{thm-restate}
\usepackage[square,numbers]{natbib}
\usepackage{color}
\usepackage{thmtools}
\usepackage{thm-restate}
\usepackage{comment}
\usepackage{enumerate}
\usepackage{hyperref}
\hypersetup{
    colorlinks=false,
    linkcolor=blue,
    filecolor=magenta,      
    urlcolor=cyan,
}

\newcommand{\AutoAdjust}[3]{\mathchoice{ \left #1 #2  \right #3}{#1 #2 #3}{#1 #2 #3}{#1 #2 #3} }
\newcommand{\Xcomment}[1]{{}}

\newcommand{\InBrackets}[1]{\AutoAdjust{[}{#1}{]}}
\newcommand{\Ex}[2][]{\operatorname{\mathbf E}_{#1}\InBrackets{#2}}
\newcommand{\Prx}[2][]{\operatorname{\mathbf{Pr}}_{#1}\InBrackets{#2}}

\newcommand{\dd}{\mathrm{d}}  

\newcommand{\eps}{\epsilon}
\newcommand{\OPT}{\mathrm{OPT}}
\newcommand{\Dis}{D^S} 
\newcommand{\MAX}{\mathrm{MAX}}
\newcommand{\MAXV}{\textsc{MaxV}}

\renewcommand{\P}{\mathbb{P}}
\newcommand{\E}{\mathbf{E}}

\newcommand{\DX}{v^S}
\newcommand{\Pandora}{PNOI } 
\newcommand{\TPandora}{T-PNOI } 

\newcommand{\Sdiscretized}{$S$-discretized}
\newcommand{\TState}{I_T}
\newcommand{\Ins}{B}
\newcommand{\DIns}{B^S}

\newcommand{\SG}{G_S}

\newcommand{\UP}{D^L}
\newcommand{\ActSet}{A}

\newcommand{\fa}{\mathrm{Fa}}
\newcommand{\SUB}{Q}

\newcommand{\End}{\mathsf{end}}

\newcommand{\PI}{{W}}

\newcommand{\poly}{\mathrm{poly}}

\newcommand{\Ut}{\mathrm{Utility}}
\newcommand{\Ls}{\mathrm{Loss}}
\newcommand{\Tbf}{T_{\textnormal{Bef}}}
\newcommand{\Tat}{T_{\textnormal{Aft}}}

\newtheorem{thm}{Theorem}[section]
\newtheorem{proposition}[thm]{Proposition}
\newtheorem{lemma}[thm]{Lemma}

\newtheorem{claim}[thm]{Claim}
\newtheorem{corollary}[thm]{Corollary}
\newtheorem{definition}[thm]{Definition}
\newtheorem{fact}[thm]{Fact}

\DeclareMathOperator*{\argmax}{arg\,max}
\DeclareMathOperator*{\supp}{supp}
\DeclareRobustCommand\iff{\;\Longleftrightarrow\;}

\newcommand{\noaccents}[1]{#1}
\newcommand{\newagentvar}[3][\noaccents]{%
\expandafter\newcommand\expandafter{\csname #2\endcsname}{#1{#3}}%
\expandafter\newcommand\expandafter{\csname #2s\endcsname}{#1{\boldsymbol{#3}}}%
\expandafter\newcommand\expandafter{\csname #2smi\endcsname}[1][i]{#1{\boldsymbol{#3}}_{-##1}}%
\expandafter\newcommand\expandafter{\csname #2i\endcsname}[1][i]{#1{#3}_{##1}}%
\expandafter\newcommand\expandafter{\csname #2ith\endcsname}[1][i]{#1{#3}_{(##1)}}%
}

\newcommand{\newvecagentvar}[3][\noaccents]{%
\expandafter\newcommand\expandafter{\csname #2\endcsname}{#1{\boldsymbol{#3}}}%
\expandafter\newcommand\expandafter{\csname #2s\endcsname}{#1{\boldsymbol{#3}}}%
\expandafter\newcommand\expandafter{\csname #2smi\endcsname}[1][i]{#1{\boldsymbol{#3}}_{-##1}}%
\expandafter\newcommand\expandafter{\csname #2i\endcsname}[1][i]{#1{\boldsymbol{#3}}_{##1}}%
\expandafter\newcommand\expandafter{\csname #2ith\endcsname}[1][i]{#1{#3}_{(##1)}}%
}

\newagentvar{alloc}{x}
\newagentvar{pay}{p}
\newagentvar{val}{v}
\newagentvar{util}{u}
\newagentvar{payment}{p}
\newagentvar{dist}{F}

\newcommand{\gittins}{\tau}
\newcommand{\amorval}{\kappa}
\newcommand{\alg}{\mathcal{A}}

\newcommand{\Dalg}{\mathcal{A}^S}
\newcommand{\threshold}{\theta}
\newcommand{\classnp}{\mathsf {NP}}
\newcommand{\classpspace}{\mathsf {PSPACE}}
\newcommand{\classp}{\mathsf {P}}
\newcommand{\leV}{\le_V}

\newcommand{\image}{R}

\newcommand{\jiawei}[1]{}
\newcommand{\daogao}[1]{}
\newcommand{\hufu}[1]{}

\title{Pandora Box Problem with Nonobligatory Inspection: Hardness and Approximation Scheme}

\author{Hu Fu
\thanks{ITCS, Shanghai University of Finance and Economics. \texttt{fuhu@mail.shufe.edu.cn}}
\and Jiawei Li 
\thanks{University of Texas at Austin. 
\texttt{davidlee@cs.utexas.edu}}
\and Daogao Liu 
\thanks{University of Washington. \texttt{dgliu@uw.edu}}
}

\date{}
\usepackage{graphicx}

\begin{document}

\maketitle

\begin{abstract}
Weitzman (1979) introduced the Pandora Box problem as a model for sequential search with inspection costs, and gave an elegant index-based policy that attains provably optimal expected payoff.
In various scenarios, the searching agent may select an option without making a costly inspection.
The variant of the Pandora box problem with non-obligatory inspection has attracted interest from both economics and algorithms researchers.
Various simple algorithms have proved suboptimal, with the best known 0.8-approximation algorithm due to Guha et al. (2008).
No hardness result for the problem was known.

In this work, we show that it is $\classnp$-hard to compute an optimal policy for Pandora's problem with nonobligatory inspection.
We also give a polynomial-time approximation scheme (PTAS) that computes policies with an expected payoff at least $(1 - \epsilon)$-fraction of the optimal, for arbitrarily small $\epsilon > 0$.
On the side, we show the decision version of the problem to be in $\classnp$.
\end{abstract}

\section{Introduction}
\label{sec:intro}

\citet{weitzman1979optimal} introduced the Pandora Box problem in 1979 as a model for sequential search with inspection costs.
An agent is to select from $n$~options, viewed as locked boxes. 
Each box~$i$ contains a value~$\vali$ drawn independently from a known distribution~$\dist_i$, but $\vali$ is revealed only if the agent opens box~$i$, incurring a search cost~$c_i$.
At any step, the agent may choose to either select a box that has been opened and quit, or to open another box.
Her goal is to maximize the expected value of the box selected, minus the search costs accrued along the way.

A policy for such a stochastic sequential problem may conceivably be adaptive in intricate ways.
It may therefore come as a surprise that \citeauthor{weitzman1979optimal} showed the problem of admitting a simple and elegant optimal policy:
there are indices, one for each box~$i$, computable from $\disti$ and~$c_i$, such that a ranking policy based on these indices maximizes the expected payoff.  

Weitzman's formulation and the index-based policy have been highly influential, and serve as the basis for many models that involve search frictions (see e.g.\@ \citet{armstrong17survey} for a survey, and \citet{mirrokni20pandora} for a recent example).  
It was later recognized that the index-based policy was a special case of \citet{gittins79banditprocesses}'s optimal algorithm for Bayesian bandits, an algorithm important in its own right, with many applications.

In \citet{weitzman1979optimal}'s motivating scenario, one searches for a technology among various alternatives; to be able to adopt any technology, research expenditure (the search cost) must be incurred before one sees the technology's benefit.
In many other scenarios, however, it is more natural to allow the agent the possibility to select an option without inspecting it.
For example, in a wireless system with multiple channels, a user intending to transmit a control packet may spend energy and time probing the transmission states of channels, but may also decide to use a channel without knowing its exact state \citep{GMS08}; 
a student making a school choice may not always pay a campus visit when she is confident enough that one choice is superior \citep{doval18}.
This problem variant is known as \emph{Pandora's problem with non-obligatory inspection} (PNOI).
In contrast to the original Pandora Box problem, various simple ranking policies are not optimal \citep{doval18}; in fact, optimal policies may be truly adaptive, in the sense that the order in which two options are inspected should depend on the outcome of the inspection of a third option.
One may easily attain at least $\tfrac 1 2$ of the optimal payoff by using the better of two simple policies: 
(i) Weitzman's index-based policy, and (ii) selecting the box with the highest expected value without any inspection.
Nontrivial algorithms with better approximation ratios have been proposed \citep{GMS08, beyhaghi2019pandora}, the best approximation ratio known so far being $0.8$ \citep{GMS08}.
On the other hand, it has been unresolved whether computing optimal policies is intractable --- the problem could be anywhere between $\classp$ and $\classpspace$-complete \citep{beyhaghi2019pandora}.  

In this work, we show that PNOI is $\classnp$-hard, giving the first hardness result for the problem (Theorem~\ref{thm:hardness}).
We also give a polynomial time approximation scheme (PTAS) for PNOI (Theorem~\ref{thm:general_ptas}).
On the side, we show the decision version of PNOI to be in $\classnp$ (Corollary~\ref{cor:NP}).

\paragraph{Computational Hardness for PNOI}

Before discussing the main idea of our hardness result,  it is helpful to first relate a structure theorem on optimal policies that we strengthen from \citet{GMS08}, 
The structure theorem by \citep{GMS08} is crucial for the $0.8$-approximation algorithm, and shows the existence of an optimal policy for which a unique box is possibly taken without inspection.
We follow their arguments a step further, and show the existence of an optimal policy~$\alg^*$ with a simple description (Theorem~\ref{thm:structure}): 
$\alg^*$ commits to a subset $T^*$ of boxes, an ordering $\sigma$ on~$T^*$, and a threshold $V_i$ for each box $i \in T^*$; $\alg^*$ opens the boxes in~$T^*$ in the ordering~$\sigma$, until either (a) a box~$i$ yields a value at least~$V_i$, at which point $\alg^*$ switches to running the index-based policy on the rest of the boxes (taking the best value seen so far as a free outside option), or (b) none of the first $|T^*| - 1$ boxes yield values passing their thresholds, at which point $\alg^*$ takes the last box in~$T^*$ without inspection.
As special cases, if $T^* = \emptyset$, $\alg^*$ is Weitzman's index policy; if $|T^*| = 1$, $\alg^*$ takes the box in~$T^*$ without inspection.
Our strengthened structure theorem shows the existence of a succinctly representable optimal policy, and implies the decision version of PNOI to be in~$\classnp$.
To take another interesting perspective, \citet{doval18} observed that, in an optimal policy for a PNOI instance, the outcome from probing a box may influence the order in which the other boxes are probed.  
Our structure theorem shows that, such order switching follows a simple structure: only during the stage of probing boxes in~$T^*$ can the order of future queries change (from~$\sigma$), and that change is triggered only when a value high enough is revealed, at which point the policy switches to the index policy.

To show hardness, we study a family of PNOI instances where it is easy to determine $T^*$ and the thresholds $V_1, \cdots, V_n$ in the above description of the optimal policy, so that the hardness is solely from deciding the ordering~$\sigma$.
In these instances, each value distribution $\disti$ is supported on $\{0, \tfrac 1 2, 1\}$, with expectation less than $\tfrac 1 2$, and index (in \citeauthor{weitzman1979optimal}'s sense) at least $\tfrac 1 2$.  
It can be shown that for such instances an optimal policy must be of the form stated in the structure theorem, with $T^*$ being the set of all boxes, and $V_i = \frac 1 2$ for each~$i$.

The possibility of switching to the index policy after each inspection adds difficulty to calculating the policy's expected payoff, but this calculation is necessary for a reduction.
Key to our analysis is to observe that a closely related policy has the \emph{non-exposed} property, a property that was first crystallized by Kleinberg, Waggoner, and Weyl (\citeyear{kleinberg2016descending}) and has been instrumental in several works in optimal search (e.g., \citep{singla18poi}; \citep{gamlath19matching}; \citep{beyhaghi2019pandora}).  
This property allows us to derive a relatively clean expression for the \emph{difference} between a policy's expected payoff and that of \citeauthor{weitzman1979optimal}'s index-based policy (Lemma~\ref{lm:utility_permutation}).
Computing an optimal policy boils down to finding an ordering~$\sigma$ that maximizes this difference.

Finally, we give a fairly technical reduction from the classical Partition problem: given a set~$S$ of $n$ positive integers, decide whether they can be partitioned into two subsets with equal sums.
We embed the $n$ integers in the parameters of $n$ boxes, and add two auxiliary boxes, $B_{n+1}$ and~$B_{n+2}$.
Box $B_{n+2}$ has both high index and high cost, so designed that $B_{n+2}$ is the unique box possibly selected without inspection, but is the first to be inspected if a value~$\tfrac 1 2$ is found.
This creates an exquisite balance between, on one hand, the saving in cost when a high-cost box is selected without inspection, and, on the other, the motive to inspect a high-index box early on.
The time point at which to switch to the index-based policy is affected by the position of~$B_{n+1}$.
We are able to set the parameters so that the most balanced partition of~$S$ is realized in the ordering~$\sigma$ of the optimal policy: the boxes before $B_{n+1}$ and those after form the partition.
The reduction is fairly involved technically due to the need for various approximations --- the expected payoff even for such simple instances of PNOI is still complex, and takes a fair amount of massaging to have terms bearing resemblance to the sums in the Partition problem.

\paragraph{Polynomial Time Approximation Scheme (PTAS)}


Our PTAS is built on a framework by \citet{fu2018ptas} that gives PTAS for a broad class of stochastic sequential optimization problems.
The main idea of the framework, to put it very roughly, is to start by considering systems with only an $O(1)$ number of possible states. 
For such systems, one can show that, in the decision tree of a policy, nodes may be grouped into a small number of blocks --- within each block, the system's state remains the same and the ordering of the actions matters little for the eventual objective.
One may therefore use dynamic programming to exhaustively optimize over decision trees consisting of such blocks, with a loss of only a small fraction of the payoff.
A natural way to cast PNOI in this framework is to let the state of the system be the highest value revealed so far.
Further manipulations enable us to inherit the main theorem of \citeauthor{fu2018ptas} and to obtain a PTAS for PNOI when there are only $O(1)$ possible values (Proposition~\ref{thm:const}).

To generalize from these restricted instances, it is natural to consider discretizing the values before applying the framework.  
For $\eps > 0$, standard discretization can reduce the support size of value distributions to $\poly(1 / \eps)$ and lose $O(\eps)$ fraction of the payoff \emph{if all values are within $\poly( 1 / \eps)$ factor of the optimal expected payoff}.  
Let's call a value \emph{large} if it is at least $\frac 1 \eps$ times the optimal payoff.
If one simply ignores large values, a sizable fraction of the payoff may be lost.
A major technical contribution of ours is a separate method to discretize large values.  
We observe that, once a large value is found, the expected additional payoff is upper-bounded and approximated by a well-behaved additive function.
We use this function to define $O(1 / \eps)$ discretization points to which we round up the large values.
These discretization points are potentially far apart, and we cannot round the values in the usual way, as that again may introduce too much error.
Crucially, we only round up values yielded by opened boxes, i.e., overestimate what one currently has, but do not round up values of boxes to be opened.
In other words, in the discretized problem, when a box yields a large value~$v$, our payoff is $v$ minus the current highest value (if $v$ is larger) and the search cost, but then pretends from this point on that the highest value is the discretized~$v$, i.e., the smallest discretization point larger than~$v$.
We show that, this non-standard discretization controls the error introduced in the payoff, and still supports optimization within \citeauthor{fu2018ptas}'s framework.

\subsection{Additional Related Works}
\label{sec:related}

In the context (and disguise) of channel probing in wireless systems, \citet{GMS08} gave a 0.8-approximation algorithm for PNOI, based on a structure theorem they showed for optimal policies.
Since the work predated the rediscovery of the Pandora Box problem in the computer science community, and was not presented as part of this line of work, it has remained little known to this community.
\citet{doval18} revived the problem in the economics literature, and observed the complex behaviors of optimal policies for PNOI.
\citet{beyhaghi2019pandora} reintroduced the algorithmic question to the econ-CS literature, and drew a connection to the adaptivity gap in stochastic submodular function maximization, which yields a simple $(1 - 1/e)$-approximation algorithm.
Through personal communications, we learned that \citet{BC22} independently obtained a PTAS for PNOI; they also strengthened \citeauthor{GMS08}'s structure theorem to a version equivalent to our Theorem~\ref{thm:structure}.

Algorithms for natural variants of the Pandora box problem have received much attention lately.
To name a few examples, \citet{leonardi20order} studied the optimal search problem when there are constraints on the order in which the boxes may be inspected;
\citet{chawla20correlated} studied the case when values in the boxes are correlated;
\citet{fu2018ptas} and \citet{segev21eptas} used \emph{Pandora box problem with commitment} as an application for their frameworks of designing PTAS and EPTAS, respectively, for stochastic combinatorial optimization problems.
\citet{singla18poi} generalized the optimal search problem to settings known as \emph{Price of Information}, where the set of options that may be selected is governed by combinatorial feasibility systems;
\citet{gamlath19matching} and \citet{fu21matching} studied such settings when the feasibility systems are given by matchings in graphs.

Hardness for computing online optimal policies in Bayesian selection problems is relatively sparse in the literature, but has been gaining attention recently.
\citet{agrawal20ordering} showed $\classnp$-hardness for choosing the optimal ordering in an online stopping problem.
\citet{papadimitriou21matching} showed $\classpspace$-hardness for the online stochastic bipartite matching problem.

\section{Preliminaries}
\label{sec:prelim}

In an instance of the classical Pandora box problem, we are given $n$ sealed boxes, each box~$i$ labeled with a search cost~$c_i \geq 0$ and a distribution $\disti$.  
Box~$i$ contains a value $\vali \geq 0$, initially hidden, drawn independently from~$\disti$, and one may open the box at cost~$c_i$ to reveal~$\vali$.
At any point, a policy may (adaptively) choose to open a sealed box, or to take the highest value revealed so far and quit.
Upon quitting, the payoff is the value taken minus the costs incurred along the way.
Our goal is to maximize the expected payoff, where the expectation is over the values and the possible randomness in the policy.
In a problem of \emph{Pandora box with non-obligatory inspection} (PNOI), it is allowed to take a box that has not been opened, in which case the payoff is the unseen value in that box minus the costs incurred before taking the box.


\paragraph{Weitzman's Index-based Policy.}
For box~$i$, define its \emph{index} $\gittins_i$ to be the unique solution to the equation $\Ex[\vali \sim \disti] {(\vali - \gittins_i)_+} = c_i$, where $(\vali - \gittins_i)_+$ denotes $\max \{0, \vali - \gittins_i\}$.
Let $\amorval_i$ be $\min \{\vali, \gittins_i\}$.

Weitzman's index-based policy first writes on each box its index, then at each stage, if the largest positive number written on the boxes is an index, the policy opens that box and writes the value revealed in place of its index; or else the largest written number is a value, and the policy takes the box with that value and terminates.\footnote{In the degenerate case where all indices are negative to begin with, the policy does nothing and quits, getting a payoff of $0$.}

\citet{kleinberg2016descending} gave a new proof for the optimality of the index-based policy for the classical Pandora box problem.  
The proof has proved powerful and inspired multiple algorithmic works \citep[e.g.][]{singla18poi, gamlath19matching, beyhaghi2019pandora}.
Part of its power is to isolate the so-called \emph{non-exposed} property (Definition~\ref{def:non-exposed}), which we use in our reduction in Section~\ref{sec:hardness}.
For completeness, we give the proof in Appendix~\ref{sec:prelim-app}.

\begin{restatable}{thm}{thmindex}[\citealp{weitzman1979optimal, kleinberg2016descending}]
\label{thm:index}
The index-based policy maximizes the expected payoff in the classical Pandora box problem.  
Its expected payoff is $\Ex{\max_{i \in [n]} \amorval_i}$.
\end{restatable}


\begin{definition}
\label{def:non-exposed}
A policy is non-exposed if, when it opens a box~$i$ and finds $\vali > \gittins_i$, it is guaranteed to take box~$i$.
More formally, a policy is non-exposed if  $\Prx{(I_i - A_i) (\vali - \gittins_i)_+} = 0$ with probability~$1$, for each~$i$.
\end{definition}

\paragraph{Decision Trees.}
In various arguments, it is convenient to analyze policies as decision trees. 
A deterministic policy~$\alg$ on an instance $\Ins$ of \Pandora is fully described by a decision tree~$T$. 
At each node~$u$, $\alg$ chooses an \emph{action}~$a_u$, which may be opening a box, taking a box without opening it, or taking a maximum value seen so far.
The latter two categories of actions lead to a leaf signifying the end of the process; upon opening a box, the system transits probabilistically to another node depending on the value revealed.

We use $\P(\alg, \Ins)$ to denote the expected profit of policy $\alg$ on a PNOI instance~$\Ins$. 
When the instance is clear from the context, we omit the second argument and write $\P(\alg)$.
For a given instance of PNOI, we let $\OPT$ be the maximum expected profit achievable by any policy.

\paragraph{Strengthened Structure Theorem}

Key to \citet{GMS08}'s approximation algorithm is their structure theorem, which states the existence of an optimal policy for which there is a unique box possibly taken without inspection.
We strengthen the theorem and show the existence of a well-structured, succinctly describable optimal policy.
Our proof largely follows the arguments in~\citep{GMS08}.
The theorem immediately shows the decision version of PNOI to be in $\classnp$.
Our main technical results do not rely on this structure theorem, although the optimal policy's structure, which features an ordering of boxes, sheds some light on our reduction in Section~\ref{sec:hardness}.




\begin{restatable}{thm}{thmstructure}
\label{thm:struc}
\label{thm:structure}
For any PNOI instance, there is an optimal policy~$\alg$ described by a subset of boxes~$T^*$, an ordering~$\sigma$ on~$T^*$, and a threshold value~$V_i$ for each box $i \in T^*$ (except the last one according to $\sigma$).
If $T^* = \emptyset$, $\alg$ is the index-based policy.
Otherwise $\alg$ opens the boxes in~$T^*$ according to the ordering~$\sigma$ until either it sees a value at least~$V_i$ from a box~$i$, or only one box remains unopened in~$T^*$.
Once it sees a value $v_i \geq V_i$ from box $i \in T^*$, $\alg$ switches to running an index-based policy on the remaining boxes, taking the highest value seen so far as a free outside option;
if this does not happen till only one box in~$T^*$ remains unopened, $\alg$ takes that box without inspection.
\end{restatable}


\begin{corollary}
  \label{cor:NP}
  The following decision version of PNOI is in $\classnp$: given a PNOI instance and $P > 0$, decide whether there is a policy with an expected payoff at least~$P$.
\end{corollary}

\section{Hardness of PNOI}
\label{sec:hardness}

In this section, we show that computing optimal policies for PNOI is $\classnp$-hard.
Our reduction makes use of the following class of PNOI instances.




\begin{definition}
\label{def:LCLRS3}
  An instance of PNOI is a \emph{low-cost-low-return-support-3 (LCLRS3)} instance if the following conditions hold: 
\begin{enumerate}
  \item each value distribution~$\dist_i$ is supported on $\{0, \frac {1} {2}, 1\}$, with probability masses $p_i \coloneqq \Prx{\vali = 1} > 0$, $q_i \coloneqq \Prx{\vali = \frac 1 2}$, $r_i \coloneqq 1 - p_i - q_i = \Prx{\vali = 0}$;
  \item for each box~$i$, the cost $c_i > 0$, and the expected value $\Ex{\vali}= p_i + \frac{q_i}{2} < \frac{1}{2}$;
  \item for each box~$i$, the index $\gittins_i \geq \frac 1 2$, which implies $\gittins_i = 1 - \frac{c_i}{p_i}$.
\end{enumerate}
  \label{def:hard}
\end{definition}

In Section~\ref{sec:normal} we show that, optimal policies for LCLRS3 instances are particularly simple --- in the language of Theorem~\ref{thm:structure}, only the ordering $\sigma$ remains to be determined.
In Section~\ref{sec:reduction} we reduce the partition problem to computing optimal policies for LCLRS3 instances.


\begin{thm}
  It is $\classnp$-hard to compute optimal policies for LCLRS3 instances of PNOI.
  \label{thm:hardness}
\end{thm}

Before proving the theorem, we quickly remark that the value $\frac 1 2$ in the support is necessary for a hardness result.
This was already observed by \citet{GMS08}.  We provide a proof for completeness in Appendix~\ref{sec:hardness-app}.

\begin{restatable}{proposition}
  {supporttwo}\citep{GMS08}
    \label{prop:support2}
There is a polynomial-time computable optimal policy for PNOI instances where all value distributions are supported on $\{0, 1\}$.
\end{restatable}



\subsection{Normal Policies and Their Payoffs}
\label{sec:normal}

In this section, we show that optimal policies for LCLRS3 instances are of a format that we call \emph{normal} (Definition~\ref{def:normal}).
We then make use of ideas from \citet{kleinberg2016descending}'s proof for the index-based policy, and derive an expression for normal policies' payoff (Lemma~\ref{lm:utility_permutation}), which plays a crucial role in the reduction we present in Section~\ref{sec:reduction}.

\begin{definition}
\label{def:normal}
A policy $\alg$ for a LCLRS3 instance is said to be \emph{normal} if
  \begin{itemize}
    \item If $\alg$ sees 0 in the first $n - 1$ boxes it opens, then $\alg$ takes the last box without inspection; this is the only situation in which $\alg$ exercises the option to bypass inspection.
    \item Whenever $\alg$ opens a box and sees value~$1$ in it, it immediately takes the box and stops.
    \item Whenever $\alg$ opens a box and sees value~$\frac{1}{2}$ in it, it forsakes the option to take a box without inspection; on the remaining boxes, $\alg$ runs the index-based policy, with an outside option of value $\frac 1 2$.  
  \end{itemize}
\end{definition}

In the language of Theorem~\ref{thm:structure}, a policy is normal if $T^*$ is the set of all boxes, and $V_i = \tfrac 1 2$ for each box~$i$.
The following lemma is straightforward; we relegate its proof to Appendix~\ref{sec:hardness-app}.
\begin{restatable}{lemma}{normal}
\label{lem:normal}
    For any LCLRS3 instance, an optimal policy~$\alg$ is normal.
\end{restatable}


Since both the set $T^*$ and the thresholds $V_i$'s for a normal policy are fixed, finding an optimal policy for an LCLRS3 instance amounts to finding an optimal permutation~$\sigma$.
For a given LCLRS3 instance and a policy~$\alg$ on it, recall that we denote the expected payoff of $\alg$ as~$\P(\alg)$.
The next observation is that, for a normal policy~$\alg$, the difference between $\P(\alg)$ and the payoff of the classical index-based policy admits a relatively clean expression in terms ofthe ordering~$\sigma$ of~$\alg$.
This is by considering an intermediate policy $\alg'$, whose payoff admits simplifications using ideas from \citet{kleinberg2016descending}'s proof (Theorem~\ref{thm:index}).


As in the proof of Theorem~\ref{thm:index}, define $\amorval_i \coloneqq \min \{\vali, \gittins_i\}$.

\begin{restatable}{lemma}{utilitypermutation}
\label{lm:utility_permutation}
    Given an LCLRS3 instance and a normal policy~$\alg$ for it, let $\sigma$ be the corresponding ordering.  
For box~$i$, let $T_{\sigma}(i)$ be the set of boxes ordered after box~$i$ by~$\sigma$, with Gittins indices strictly larger than $\gittins_i$; that is, $T_\sigma(i) \coloneqq \{j \in [n]: \sigma^{-1}(j) > \sigma^{-1}(i), \gittins_{j} > \gittins_i\}$.  
For $i \in [n]$ and $T \subseteq [n]$, define $g(i, T) \coloneqq \Ex{(\max_{j \in T} \amorval_j - \gittins_i)_+}$.\footnote{$g(i, \emptyset) \coloneqq 0$.}
  Let $\alg_P$ be the index-based policy on the instance.   Then 
    \begin{align}
      \P(\alg_P) = \P(\alg) + \sum_i p_{i}g(i, T_\sigma(i)) \prod_{j = 1}^{\sigma^{-1}(i)-1} r_{\sigma(j)} 
      - c_{\sigma(n)}\prod_{i = 1}^{n-1} r_{\sigma(i)}.
      \label{eq:diff-pandora}
    \end{align}
\end{restatable}


The proof can be found in Appendix~\ref{sec:hardness-app}.
Here we briefly explain the terms in~\eqref{eq:diff-pandora}.
Recall that $\P(\alg_P)=\E[\max_{j\in[n]}\amorval_j]$.
For each box~$i$, with probability $p_{i} \prod_{j = 1}^{\sigma^{-1}(i)-1} r_{\sigma(j)} $, the boxes before~$i$ all take value 0, and box $i$ takes value 1, with $\amorval_i=\gittins_i$.
Conditioning on this event, the expected gap between $\max_{j \in [n]} \amorval_j$ and $\amorval_i$
is $g(i, T_\sigma(i))$.
The advantage of the normal policy $\alg$ is that it can save the cost $c_{\sigma(n)}$ to open the last box with probability $\prod_{i = 1}^{n-1} r_{\sigma(i)}$, and this leads to the last term $- c_{\sigma(n)}\prod_{i = 1}^{n-1} r_{\sigma(i)}$.

\subsection{Reduction}
\label{sec:reduction}
In this section we give a polynomial-time reduction from the classical partition problem to PNOI.

\begin{definition}[Partition Problem]
Given a multiset $S$ of positive integers $s_1,\cdots,s_n$, 
decide whether $S$ can be partitioned into two subsets $S_1$ and $S_2$ such that the sum of the numbers in $S_1$ equals the sum of the numbers in $S_2$.
\end{definition}

It is well-known that Partition problem is $\classnp$-complete \citep[see, e.g.][]{garey1979computers}.
It is also not difficult to show that the problem is still $\classnp$-hard when $1\leq s_1\leq \cdots \leq s_n\leq 2^n$.
We assume so in the following reduction.
We first formally give the reduction, and explain the intuition below.

\paragraph{Reduction from LCLRS3 to Partition.}
Given the multiset $S = \{s_1, \ldots, s_n \}$ of integers between $1$ and~$2^n$, fix two constants $\Gamma=2^{8n}$ and $\Delta = 2^{-7n}$. 
We construct an LCLRS3 instance with $n+2$ boxes, denoted as $B_1,\cdots,B_{n+1},B_{n+2}$.

For box $B_{n+1}$, set $p_{n+1}=1/\Gamma,q_{n+1}=1 - 41/\Gamma$ and $c_{n+1}=p_{n+1}/2$. 
This makes $\tau_L \coloneqq \gittins_{n+1}=\frac{1}{2}$ and $r_{n+1}=40/\Gamma$.

For box $B_{n+2}$, set $p_{n+2}=q_{n+2}=1/8$, $c_{n+2}=1/32$ and thus $\gittins_{n+2}=3/4$. 

For each $i\in [n]$, set $p_i=q_i=s_i/\Gamma$. 
Set a constant $\gittins_H = 3/4 - O(\Delta)$, whose precise value is to be determined later (see Claim~\ref{clm:setting_t_tau_H}). 
Set $c_i$ to make
$$\tau_i = \tau_H + \frac{p_i p_{n+1} (1-p_{n+2})(\tau_H - \tau_L)}{2p_{n+2}} = \tau_H + O(\Delta^2).$$

Note that $p_i \leq \Delta$ for any $i \in [n+1]$, since we assumed $s_i \leq 2^n$.
The construction ensures $\tau_{n+2}>\tau_i > \gittins_H >\tau_L = 1/2$ for any $i \in [n]$.

\vskip 25pt

\begin{figure}[ht!]
		\centering
		\includegraphics[scale=0.65]{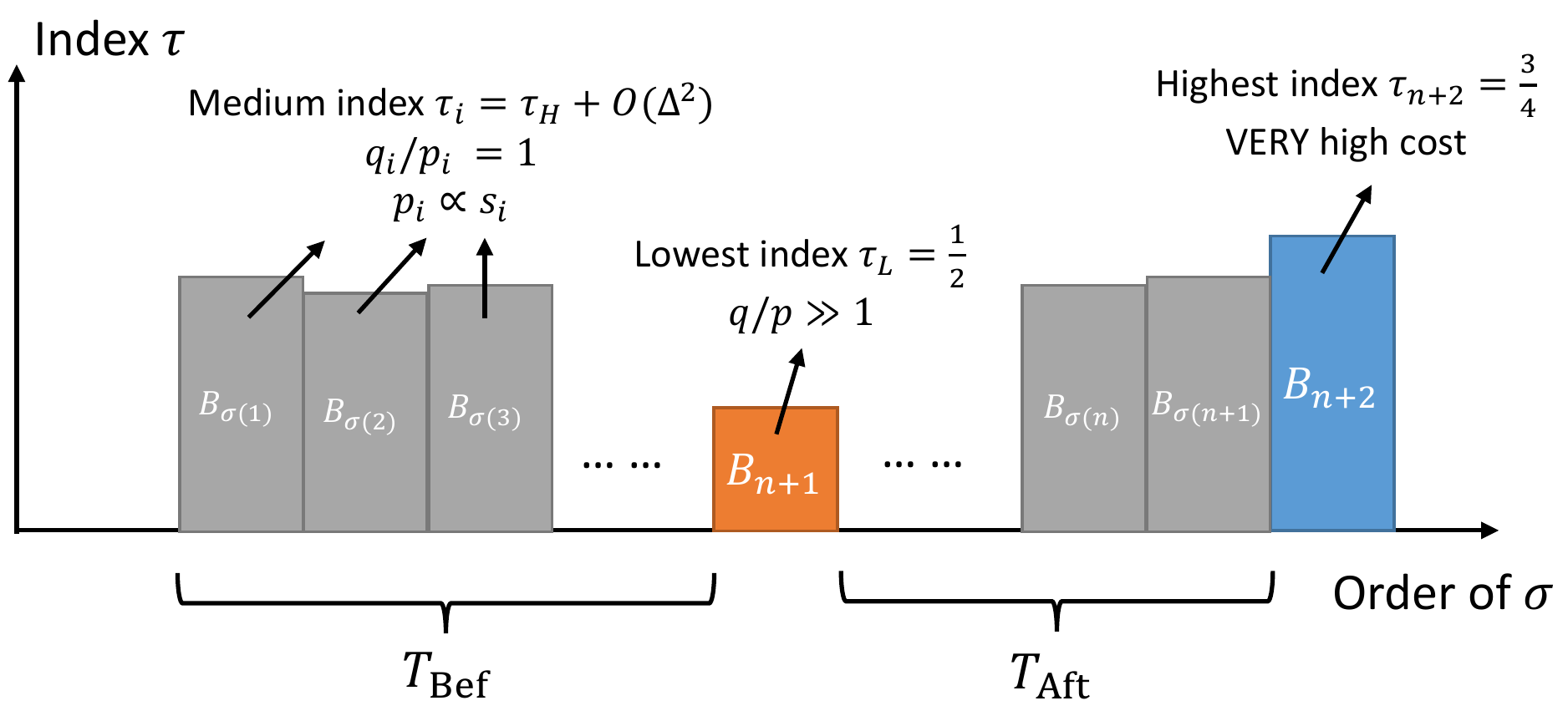}
		\caption{An overview of our reduction. Rectangles represent boxes, shown from left to right in the order of $\sigma$. The height of each rectangle represents the box's index.}
		\label{fig:partition}
	\end{figure}
Recall that the optimal solution to any LCLRS3 instance can be represented by a permutation $\sigma$.
Given the permutation $\sigma$, the position of $B_{n+1}$ in the permutation plays a crucial role in the following analysis.
Let the position of $B_{n+1}$ in~$\sigma$ be $\xi$, i.e. $\sigma(\xi)=n+1$.
Thus $B_1, \cdots, B_n$
are partitioned into two sets in~$\sigma$: those before~$B_{n+1}$ and those after, which we denote by $\Tbf$ and $\Tat$, respectively.
Formally, $\Tbf \coloneqq \{i: 1\leq i\leq n, \sigma^{-1}(i)< \xi\}$ and $\Tat \coloneqq \{i: 1\leq i\leq n, \sigma^{-1}(i)>\xi\}$.
The next key lemma builds the bridge between the partition problem and LCLRS3 instances:
\begin{lemma}
\label{lm:reduction}
The answer to the Partition problem with input~$S$ is \textsc{Yes} if and only if $\sum_{i\in\Tat^*}p_i=\sum_{i\in\Tbf^*}p_i$ in the permutation $\sigma^*$ that corresponds to an optimal policy for the LCLRS3 instance.
\end{lemma}

Theorem~\ref{thm:hardness} follows immediately from Lemma~\ref{lm:reduction} and the fact that Partition problem is NP-hard.
It remains to prove Lemma~\ref{lm:reduction}.

\paragraph{Intuition of the Reduction and Proof Overview.}
By Lemma~\ref{lem:normal} and Lemma~\ref{lm:utility_permutation}, giving an optimal policy for the LCLRS3 instance boils down to finding a permutation $\sigma$ that maximizes the objective value
\begin{align}
\label{eq:object}
    \Ut(\sigma):=
    c_{\sigma(n+2)}\prod_{i = 1}^{n+1} r_{\sigma(i)}-\sum_i p_{i} g(i, T_\sigma(i)) \prod_{j = 1}^{\sigma^{-1}(i)-1} r_{\sigma(j)}.
\end{align}
 
We would like to focus on the more complex second term (which we call the loss term \eqref{eq:loss_sigma_original}).
The role of box~$n+2$ is to fix the first term: $c_{n+2}$ is a constant whereas $c_1, \ldots, c_{n+1}$ are exponentially small.
$c_{n+2}$ is so large compared with all other terms in~\eqref{eq:object} that
any reasonable policy must leave box~$n+2$ till the end:

\begin{restatable}{claim}{claimFixFinalBox}
\label{clm:fix_final_box}

Let $\sigma^*$ be a permutation which maximizes~\eqref{eq:object}.  Then
    $\sigma^*(n+2)={n+2}$.
\end{restatable}

From the discussion above, an optimal policy must leave box~$n+2$ to the last in its permutation~$\sigma$, and the ordering of the other boxes must minimize the loss term:
\begin{align}
\label{eq:loss_sigma_original}
    \Ls(\sigma) \coloneqq \sum_i p_{i} g(i, T_\sigma(i)) \prod_{j = 1}^{\sigma^{-1}(i)-1} r_{\sigma(j)}.
\end{align}

There is a non-trivial trade-off for deciding the position of $B_{n+1}$.
On the one hand, it can be shown that if all the $n+1$ boxes before $B_{n+2}$ have the same index~$\tau$, the ones with higher ratios of $q_i / p_i$ should be opened early to maximize the utility, i.e., the box with a higher ratio of $q_i/p_i$ should be put earlier in the permutation.
On the other hand,
if the ratio $q_i / p_i$ is the same for each of these $n+1$ boxes, the box with a higher index should be opened earlier.
 The special box $B_{n+1}$ has a smaller index than boxes $B_1, \ldots, B_n$, but a much higher ratio of $q_{n+1} / p_{n+1}$. 
 Finding the non-trivial trade-off in deciding the position of $B_{n+1}$ in~$\sigma$ can be shown $\classnp$-hard.

\subsection{Correctness of Reduction: A Sketch}
\label{sec:reduction-correct}
We demonstrate all the statements and prove the key Lemma~\ref{lm:reduction}.
The omitted proofs can be found in Appendix~\ref{sec:proof_for_sketch}.
We show the key lemma by using a function with $\sum_{i\in \Tat}p_i+p_i^2$ as the single variable to approximate Equation~(\ref{eq:loss_sigma_original}). 
For ease of notation, define  \begin{align*}
y \coloneqq \sum_{i\in S} \frac {s_i}{\Gamma}+ \left(\frac {s_i}\Gamma \right)^2=\sum_{i\in\Tat\cup\Tbf}p_i+p_i^2, \quad 
x \coloneqq \sum_{i\in \Tbf}p_i+p_i^2.
\end{align*}
Note that $y$ is fixed once $S$ is given, whereas $x$ is a function of $\Tbf$ and hence of~$\sigma$.


\begin{restatable}{lemma}{lemmaApproxLossByh}
\label{clm:approx_loss_by_h}
The parameters of the instance can be set up so that 
\begin{align}
\label{eq:approx_loss_by_h}
h(x)-O(n^2\Delta^4)
   & \leq \frac{\Ls(\sigma)-C 
   \pm O(n^2\Delta^4)}{k_1} 
   \leq  
   h(x) + O(n\Delta^3), \\
   \frac{k_2}{k_1} & = 2e^{y/2} \pm O(\Delta^2), \nonumber
\end{align}
where 
\begin{align}
\label{eq:h}
    h(x):=e^{-2x}\left(1-\frac{k_2}{k_1}e^{-y+x}\right),
\end{align}
with $C,k_1,k_2$ as constants independent of $\sigma$: 
\begin{align}
\label{eq:k1}
    k_1:=- \frac 1 2  {p_{n+2}(\tau_{n+2}-\tau_H)(p_{n+1}+q_{n+1})} + p_{n+1}[(1-p_{n+2})(\tau_H-\tau_L)+p_{n+2}(\tau_{n+2}-\tau_L)],
\end{align}
\begin{align}
\label{eq:k2}
    k_2:=p_{n+1}(1-p_{n+2})(\tau_{H}-\tau_L),
\end{align}
\begin{align}
\label{eq:C}
C:=\frac 1 2 {p_{n+2}(\tau_{n+2}-\tau_H)} \left(1-\prod_{i\in[n+1]}r_i \right) + \frac 1 2 k_2 \sum_{i = 1}^n p_i^2.
\end{align}
\end{restatable}

We give a road map for the proof once we have Lemma~\ref{clm:approx_loss_by_h}.
The minimum value of $h(x)$ is taken at $x^*=y-\ln(2k_1/k_2)$.
When $k_1 / k_2$ is near $2e^{y/2}$, $x^*$ is close to $y/2$.
Our goal is to have the most even partition of~$S$ be the $\Tbf$ and~$\Tat$ of an optimal policy, which in turn should have $x$ as close to $y/2$ as possible.
Even with the approximation given in Lemma~\ref{clm:approx_loss_by_h}, a few obstacles still stand in the way: $x$ is not $\sum_{i \in \Tbf} p_i$, nor is $y$ equal to $\sum_{i \in [n]} p_i$; both of them have second-order terms, which cause further distortion in the objective through the fact that $x$ and~$y$ appear in the exponents in~$h$.
We overcome these difficulties by carefully controlling the order of errors throughout our calculation:
$p_i$'s are so small that the second-order terms in $x$ and~$y$ are negligible; analytical properties of~$h$ (Claim~\ref{clm:strong_convexity}) guarantee that, around its optimum, $h$ is sensitive enough to perturbations, so that suboptimal solutions can be told from the optimal.



Much of the proof of Lemma~\ref{clm:approx_loss_by_h}, which is fairly technical, is relegated to Appendix~\ref{sec:hardness-app}.  
We mention a tool instrumental in simplifying the calculations, which also explains our setting $p_i = q_i$ for all $i \in [n]$:

\begin{restatable}{lemma}{lemmaFromScheduling}
\label{lm:from_scheduling}
Given two sequences of positive real numbers $p_1,p_2,\cdots,p_n$ and $r_0,r_1,\cdots,r_n$.
Let $r_0=1$.
If there exists a constant $c>0$ such that $p_i/(1-r_i)=c$ for each $1\le i\leq n$, then we have
\begin{align*}
    \sum_{i=1}^{n}p_i\prod_{j=0}^{i-1}r_j=c \left(1-\prod_{i=1}^{n}r_i \right) \; .
\end{align*}
\end{restatable}


After much simplification, the main terms of $\Ls(\sigma)$ are given in Claim~\ref{lm:before_approx}, before we apply analytical tools and turn products to sums in the exponent (Fact~\ref{fct:Taylor_series}, Claim~\ref{clm:approx_error}), which leads to Lemma~\ref{clm:approx_loss_by_h}.  
Note that we have to appeal to second-order approximations of the exponential function for the required precision in the proof.
The setup of the parameter $\gittins_H$ is given in Claim~\ref{clm:setting_t_tau_H}.

\begin{restatable}{claim}{claimBeforeApprox}
    \label{lm:before_approx}
    For a non-empty set $T \subseteq [n]$, let $f(T):=\prod_{i\in T}r_i=\prod_{i\in T}(1-2p_i)$ and $g(T) \coloneqq \prod_{i\in T}(1-p_i)$.
Also let $f(\emptyset) = g(\emptyset) = 1$.
Then
\begin{align}
\label{eq:Ls_of_sigma}
    \Ls(\sigma)=k_1 f(\Tbf)-k_2 f(\Tbf)g(\Tat)- k_2\sum_{i\in \Tat} p_{i}^2 / 2 + C \pm O(n^2\Delta^4).
\end{align}
\end{restatable}

\begin{restatable}{claim}{claimSettingt}
\label{clm:setting_t_tau_H}
If we choose~$t$ so that $|t-2e^{y/2}|\leq O(\Delta^2)$, and set $\gittins_H$ as follows, then $\tfrac{k_2}{k_1} = t$:
\begin{align}
\label{eq:set_tau_H}
    \tau_H=\frac{-3t\Gamma+28+94t}{-4t\Gamma+56+104t}.
\end{align}
\end{restatable}

\begin{fact}
\label{fct:Taylor_series}
For $0\leq x\leq 1/2$, we have
    $ 1-x\leq e^{-x}\leq 1-x+x^2/2$,
and $1-x \leq e^{-x-x^2/2}\leq 1-x+O(x^3).$
\end{fact}

\begin{restatable}{claim}{claimApproxError}
\label{clm:approx_error}
For any subset $T$ of the first $n$ boxes, one has 
\begin{align*}
    e^{-\sum_{i\in T}2(p_i+p_i^2)}\geq f(T)\geq e^{-\sum_{i\in T}2(p_i+p_i^2)}-O(n\Delta^3),\\
    e^{-\sum_{i\in T}(p_i+p_i^2/2)}\geq g(T)\geq e^{-\sum_{i\in T}(p_i+p_i^2/2)}-O(n\Delta^3).
\end{align*}
\end{restatable}

With the approximation in Lemma~\ref{clm:approx_loss_by_h} in hand, we are almost ready to prove Lemma~\ref{lm:reduction}.
The next lemma shows that the function $h$ is sensitive enough to perturbations around its minimum.

\begin{restatable}{claim}{claimStrongConvexity}
    \label{clm:strong_convexity}
    If $| k_2/k_1 -2e^{y/2}|\leq O(\Delta^2)$, $\eps \in \mathbb R$ is such that $2^{-6n}\geq |\epsilon|\geq 1 / \Gamma = 2^{-8n}$ , let $x^*\in[0,1/2]$ be where $h(x)$ takes its minimum value, then $|x^*-\frac{y}{2}|\leq O(\Delta^2)$,
    \begin{align*}
        h(x^*+\epsilon)\geq h(x^*)+\epsilon^2/2.
    \end{align*}
\end{restatable}

\begin{proof} [Proof of Lemma~\ref{lm:reduction}]
The ``if'' part is obvious: if the permutation~$\sigma^*$ of a policy yields a partition $\Tbf^*$ and $\Tat^*$ with $\sum_{i \in \Tbf^*}p_i = \sum_{i \in \Tat^*} p_i$, then since $p_i = s_i / \Gamma$ for each $i \in [n]$, $(\Tbf^*, \Tat^*)$ certifies that $S$ is a \textsc{Yes} instance of Partition.

For the ``only if'' part, 
suppose $S$ can be partitioned into disjoint subsets $S_1$ and~$S_2$ with $\sum_{s\in S_1}s=\sum_{s\in S_2}s$, we show that any policy whose corresponding $\Tbf$ and~$\Tat$ is not an even partition of~$S$ must be suboptimal.
By our setting of parameters and Claim~\ref{clm:setting_t_tau_H}, we know $|x^*-\frac{y}{2}|\leq O(\Delta^2)$. 
For any permutation $\sigma$ whose corresponding sets $\Tat,\Tbf$ are such that $\sum_{i\in\Tat}p_i\neq \sum_{i\in\Tbf}p_i$, we have $|x-y/2|\ge 1/\Gamma- \sum_{i\in\Tat\cup\Tbf}p_i^2\ge 1/\Gamma-n\Delta^2$, and hence $|x-x^*|\geq 1/\Gamma -n\Delta^2 / 2$.
Hence, one has
\begin{align*}
    \frac{\Ls(\sigma)-C + O(n^2\Delta^4)}{k_1}
    & \geq h(x) - O(n^2 \Delta^4) 
    > h(x^*) + \Omega(1 / \Gamma^2) \\
    &
    \ge \frac{\Ls(\sigma^*)-C 
    + O(n^2\Delta^4)}{k_1}+\Omega(1/\Gamma^2),
\end{align*}
where the second inequality follows from Claim~\ref{clm:strong_convexity}, the first and last inequality follow from Lemma~\ref{clm:approx_loss_by_h}. 
Hence $\Ls(\sigma)>\Ls(\sigma^*)$.
\end{proof}

\section{Polynomial-Time Approximation Scheme}
\label{sec:ptas}

This section gives a polynomial time approximation scheme (PTAS) for PNOI.  
The PTAS is based on \citet{fu2018ptas}'s framework which gives approximation schemes for a class of stochastic optimization problems with $O(1)$ sized state spaces.
The natural instantiation of PNOI in this framework uses values as states, and must reduce their number to $O(1)$.  
Standard discretization, however, only works when all values are small.
Our main technical contribution is a novel discretization method tailored for larger values (Section~\ref{sec:ptas:VH_discrete}).

\subsection{Fu et al.'s Framework and Its Application to PNOI}\label{sec:ptas:prob}
\label{sec:SSDP}
    
    We first introduce \citet{fu2018ptas}'s framework.
    We make some simplifications in this presentation, referring the reader to the original paper for full details.
    We then adapt the framework to PNOI and obtain a PTAS when the support of the value distributions is of $O(1)$ size; Section~\ref{sec:ptas:discretization} deals with the general case.

    \paragraph{\citeauthor{fu2018ptas}'s SSDP framework.}
    A \emph{stochastic sequential decision process} (SSDP) is given by a 5-tuple $(V, \ActSet, f, G, I_0)$, 
    where $V$ is the set of states of the process,
    $I_0 \in V$ the initial state, 
    $\ActSet$ the set of actions,
    $f: V \times \ActSet \to V$ the stochastic state transition function, and 
    $G: V \times \ActSet \to \mathbb R$ the marginal payoff function. 
    In each step~$t$, a policy chooses an action $a \in \ActSet$, with the restriction that each action can be taken at most once during the process; a policy may also choose to end the process at any time. 
    Let $I_t$ denote the state at round $t$.
    Then when action $a_t \in \ActSet$ is taken at state $I_t \in V$ in round~$t$, the state becomes $I_{t+1} = f(I_t, a_t)$, producing a marginal payoff of $G(I_t, a_t)$ whose expectation $\Ex{G(I_t,a_t)}$ is non-negative.
    Generally, both $f$ and~$G$ are random, and may be correlated.
    In execution of the process, the total payoff is the sum of the marginal payoffs incurred in all rounds; 
    our goal is to design a policy that maximizes the expected total payoff. 
    \footnote{\citeauthor{fu2018ptas}'s original framework allows the total payoff to also include a payoff that depends on the final state reached.  Such final payoffs can be easily simulated by marginal payoffs, and our adoption of the framework for PNOI calls for no such final payoffs.}
    
    Given an SSDP, let $\OPT$ be the maximum expected payoff obtainable by any policy, and let $\MAX$ be the maximum expected payoff if a policy may (hypothetically) choose to start with any state $I'_0 \in V$, in place of $I_0$.
    
    It is not difficult to see that, for any randomized policy~$\alg^R$, 
    there is a deterministic policy whose expected payoff is no less than that of~$\alg^R$. 
    Therefore we focus on deterministic policies.
    \begin{restatable}{thm}{theoremJian}[Essentially from \cite{fu2018ptas}]\label{thm:jian}
        If an SSDP problem $(V, \ActSet, f, G, I_0)$ satisfies the following conditions:
        \begin{enumerate}
            \item The number of possible states is a constant, i.e., $\left| V \right| = O(1)$.
            \item The state space $V$ admits an ordering ``$\geq$'' such that $f(I, a) \geq I$ for any $I\in V, a\in \ActSet$, i.e., the state is non-decreasing with probability~$1$.
            \item There exists an optimal policy that never takes an action with a negative expected marginal payoff in any round.
            \item $\MAX = O(\OPT)$.
        \end{enumerate}
    
        Then, for any fixed $\eps>0$, there is a policy $\alg$ computable in time $n^{2^{O(\eps^{-3})}}$, with expected payoff at least $(1-\eps) \cdot \OPT$.
    \end{restatable}
    
    We sketch the main ideas behind Theorem~\ref{thm:jian} in Appendix~\ref{sec:ptas:app_tree}; readers interested in the details are referred to \citet{fu2018ptas}. 

    
    \paragraph{\Pandora as an SSDP.}
    
    We now present \Pandora in the framework of SSDP and apply Theorem~\ref{thm:jian} to obtain a PTAS for \Pandora when the value distributions have small supports.
    
    \begin{proposition}\label{thm:const}
        If 
        $|\bigcup_i \supp(F_i)| = O(1)$, then for any fixed $\eps>0$, there is a polynomial-time algorithm that computes a policy with an expected payoff at least $(1-\eps)\cdot \OPT$.
    \end{proposition}
    
        We prove the proposition by casting the PNOI problem as an SSDP and then applying Theorem~\ref{thm:jian}.
        In doing so, we must satisfy the conditions of Theorem~\ref{thm:jian}.
        
        Let $V$, the set of states, be $\bigcup_i \supp(F_i)$.
        The ``$\geq$'' ordering in condition 2.\@ is simply the natural ordering on reals.
        Let $\MAXV$ denote the maximum value in~$V$.
        The action space $\ActSet$ consists of three parts: 
        $A_0 = \{a_1^0, \ldots, a_n^0\}$, where $a_i^0$ is the action of opening box~$i$;
        $A_1 = \{a_1^1, \ldots, a_n^1\}$, where $a_i^1$ is the action of taking box~$i$ without opening it (and ending the process);
        and $\End$, the action of taking the maximum value seen so far (and ending the process).
        Note that, for each~$i$, at most one of the actions $a_i^0, a_i^1$ could be chosen throughout the process, and an action in~$A_1$ precludes the action~$\End$.\footnote{In the application of Theorem~\ref{thm:jian}, 
        such constraints 
        are easily taken care of by the dynamic programming.  Similar constraints arose in \emph{Committed ProbeTop-$k$ Problem} and \emph{Committed Pandora Problem}, and were handled similarly by \citet{fu2018ptas}.} 
        Set $I_0 = 0$ and let the state~$I_t$ be the largest value seen in the boxes opened in the first $t$ rounds. 
     For any $I \in V$ and $i \in [n]$, the state transition function $f$ is defined as
    \begin{align*}
             f(I, a_i^0) & = \max\{I, \val_i\},  \\
             f(I, a_i^1) & = \max\{I, \Ex{\vali}\},
             \\
             f(I, \End) & = I;
    \end{align*}
   the marginal payoff function $G$ is
        \begin{align*}
            G(I, a_i^0) &= (\max\{I, \val_i\} - I) - c_i, \\
            G(I, a_i^1) &= \Ex{\val_i} - I, \\
            G(I, \End) &= 0.
    \end{align*}
    
    \begin{claim}
    $G(I, a)$ is non-increasing in~$I$ for any~$a$, and is Lipschitz in~$I$, i.e., for any $I_1 < I_2$ and any action~$a$, $G(I_2, a) - G(I_1, a) \leq I_2 - I_1$.
    \end{claim}
    
     It is straightforward to see that, in any execution of a reasonable policy,\footnote{Reasonable here means the policy never takes a box~$i$ without inspection if $\Ex{\vali}$ is smaller than the current state~$I$.} the sum of marginal payoffs from all rounds is exactly the payoff in the \Pandora problem.
    
    We now check the conditions 
    of Theorem~\ref{thm:jian}.
    The first two are satisfied immediately.  
    For the last condition, note that,  changing $I_0$ to any positive value only decreases the marginal payoffs (because $G$ is non-increasing in~$I$), so $\MAX = \OPT$. Condition 3 is guaranteed by the lemma below. 
    
    \begin{restatable}{lemma}{lemmaBad}\label{lem:bad}
        An optimal \Pandora policy never takes an action with a negative expected marginal payoff in any round.
    \end{restatable}

    The argument is most readily seen using the language of decision trees and formally proved in Appendix~\ref{sec:ptas:app_tree}.
        We have thus shown that all four conditions in Theorem~\ref{thm:jian} are satisfied.   Proposition~\ref{thm:const} therefore follows.
        

    
    
        

    \subsection{Discretization and PTAS for General PNOI}
    \label{sec:ptas:discretization}
    
    The condition $\left| V \right| = O(1)$ is essential to Proposition~\ref{thm:const}. 
    To obtain a PTAS for general \Pandora instances, we look to discretize values to reduce the state space size.   
    However, standard discretization turns out to work only for values small compared to $\OPT$ (Section~\ref{sec:ptas:LV-discrete}).  
    We develop in Section~\ref{sec:ptas:VH_discrete} a separate, novel technique to handle large values.
    
    \begin{restatable}{thm}{theoremGeneralPTAS}\label{thm:general_ptas}
            For any fixed constant $\eps>0$, there is a polynomial-time algorithm for \Pandora that computes a policy with an expected payoff at least $(1-O(\eps))\cdot \OPT$.
    \end{restatable}

      
	Throughout this section, we fix a threshold $\threshold \in [\OPT / \eps, 2\OPT / \eps]$, which can be obtained in polynomial time by running the simple approximation algorithm mentioned in the Introduction.
    \emph{Small values} refer to values at most~$\threshold$, and \emph{large values} are those above~$\threshold$. 
    Proofs omitted in this section can be found in Appendix~\ref{sec:ptas:app_discret}.
        

    \subsubsection{Discretization of Small Values}
    \label{sec:ptas:LV-discrete}
    
    \Xcomment{
        \begin{restatable}{thm}{theoremDiscret}\label{thm:discret}
            For any fixed $\eps>0$, there is a polynomial-time computable policy~$\alg$ with expected payoff at least $(1-\eps)\cdot \OPT$ for \Pandora instances that satisfy at least one of the following two conditions:
            \begin{enumerate}
                \item there is a constant $c_{\eps}$, which only depends on $\eps$, such that $\val_i \leq c_\eps \cdot \OPT$ with probability 1, for all~$i$; 
                \item there is a constant $c_{\eps}$, which only depends on $\eps$, such that each value distribution $\dist_i$ is discrete and, for any value $v$ in the support, has $\dist_i(v) \geq 1 / c_\eps$.  
            \end{enumerate}
        \end{restatable}
    }
    
    
    

    Standard discretization techniques can handle small values.
    For fixed $\eps$ and~$\threshold$, define a discretization function 
    $\Dis: x \mapsto \lfloor \frac{x}{\threshold \cdot \eps^2} \rfloor \cdot \threshold \cdot \eps^2$. 
    We say $\Dis(\vali)$ is the \emph{\Sdiscretized} value of~$\vali$, and a policy $\Dalg$ is \emph{\Sdiscretized} if its decisions only depend on the \Sdiscretized\ values $\Dis(\vali[1]), \ldots, \Dis(\vali[n])$.
    Recall that $\P(\alg, \Ins)$ denotes the expected profit of a policy $\alg$ on a PNOI instance~$\Ins$. 
    
    \begin{restatable}{lemma}{lemmaDiscret}\label{lem:discret}
        Let $\Ins$ be an instance of PNOI and $\DIns$ the \Sdiscretized\ instance of~$\Ins$, in which each value $\val_i$ is replaced by $\Dis(\vali)$.
        Then,
        \begin{enumerate}
            \item for any policy~$\alg$, there is a \Sdiscretized\  policy $\Dalg$ such that 
            $$\P(\Dalg, \DIns) \geq  \P(\alg, \Ins) - O(\eps) \cdot \OPT;$$
            \item for any \Sdiscretized\ policy $\Dalg$, $\P(\Dalg, \Ins) \geq \P(\Dalg, \DIns)$. 
        \end{enumerate}
    \end{restatable}

    Let $\MAXV$ be the largest value in all the supports of the value distributions. 
    If $\MAXV\leq \threshold$, then the support of $\Dis(\vali[1]), \cdots, \Dis(\vali[n])$ is $O(1)$.  
    Combining Lemma~\ref{lem:discret} and Proposition~\ref{thm:const} yields a PTAS for such instances.
    Nevertheless, when $\MAXV \gg \threshold$, the support size of  \emph{\Sdiscretized} may be large, and we deal with this considerably more challenging case in Section~\ref{sec:ptas:VH_discrete}.
    
    \subsubsection{Discretization of Large Values}
    \label{sec:ptas:VH_discrete}
    
    Our discretization of large values is based on a few insights.  First, any reasonable policy, after having seen a large value, should switch to the index policy.  
    We observe that the additional expected payoff at this stage is upper bounded and approximated by an additive function (Lemma~\ref{lem:approxf}). 
    We use this function to set $O(1 / \eps)$ discretization points (Definition~\ref{def:thetas}).
    Second, we find a non-standard way to make use of these discretization points. As we explain in more detail below, simply rounding values to the closest discretization points (as in Lemma~\ref{lem:discret}) may introduce too much error.
    
        
    
    \begin{restatable}{lemma}{lemmaNoTwoJackpot}\label{lem:no2jactpot}
	  If $\alg^*$ is an optimal policy, after a probed box yields a value $v^* \geq \threshold$, $\alg^*$ only opens boxes with indices at least $v^* \geq \threshold$.  
	  Moreover, with probability at most $\eps$, any box with index at least $\threshold$ contains a value at least $\threshold$.
    \end{restatable}

        
        Let $\PI(S, v) \coloneqq \Ex{(\max_{i \in S} \amorval_i - v)_+}$ be the expected additional profit gained from the Weitzman's index-based policy on a set~$S$ of unopened boxes, when the largest revealed value is~$v$. 
        Let $F(S,v)$ be $\sum_{i \in S} \Ex{(\amorval_i - v)_+}$.  We show $F(S, v)$ upper bounds and approximates $\PI(S, v)$ for large~$v$.

        
        \begin{restatable}{lemma}{lemmaApproxF}
        \label{lem:approxf}
           For any $S \subseteq [n]$ 
           and $v \geq \threshold$, we have  $ F(S,v) \geq \PI(S, v) \geq (1- \eps) \cdot  F(S,v)
           $.
        \end{restatable}
        
        From Lemma~\ref{lem:approxf}, we know $F([n], \threshold) < 2\cdot \OPT$ when $\eps$ is small. 
        We also have $F([n], \MAXV) = 0$, and that $F([n], v)$ is non-increasing in~$v$.
	The following discretization points can be found in polynomial time:
	\begin{definition}
	  Let $m = 2/\epsilon$.  Let $\threshold = \threshold_1 < \threshold_2 < \ldots < \threshold_m = \MAXV$ be such that for any $i < m$, $F([n], \threshold_i) - F([n], \threshold_{i+1}) < \eps \cdot \OPT$. 
        The discretization function $\UP(x)$ is: for $x \in [\threshold, \MAXV]$, $\UP(x)$ is the smallest $\threshold_i$ with $\threshold_i \geq x$; for $x < \threshold$, $\UP(x) = x$.
	  \label{def:thetas}
	\end{definition}
	
	    
	    Our choice of discretization points $\threshold_1, \cdots, \threshold_m$ is motivated by controlling the error in profit introduced when a \emph{revealed, large} value~$v$ is discretized to $\UP(v)$.
	    No approximation is guaranteed if one discretizes all values, revealed and unrevealed.
	    In particular, one may significantly overestimate the payoffs if one directly discretizes the value distribution using $\threshold_1, \cdots, \threshold_m$.
	    To address this, we use raw values to calculate marginal payoffs and only discretize the revealed values.
	    In the language of SSDP, we only discretize the state transition function, but not the marginal payoff function. A PTAS is made possible by the subtle fact that Theorem~\ref{thm:jian} only requires a small state space (and has no such requirement on the payoff function).
        
        
        Now we formally give our discretization.
        A PNOI instance~$B$ is discretized to an SSDP instance~$B^L$.
	Compared with the SSDP instantiation we gave for PNOI in Section~\ref{sec:SSDP}, $B^L$ modifies the {state transition} function by rounding up newly revealed values using $\UP(\cdot)$; however, given a current (rounded) state,  the marginal payoff function of~$B^L$ is given by the raw values.
        Formally:
        \begin{enumerate}
            
            \item 
            For each state $I$ and action $a_i^0, a_i^1$,
            the state transition function $f^L$ in $B^L$ is 
                        \begin{align*}
                             f^L(I, a_i^0) &= \UP(\max\{I, \val_i\}); \\
                             f^L(I, a_i^1) &= \max \{I, \Ex{\vali} \}; \\
                             f^L(I, \End) &= I.
                        \end{align*}
            \item  For each state~$I$ and action $a_i^0, a_i^1$, the marginal payoff function $G^L$ in~$B^L$ is:
                        \begin{align*}
                             G^L(I, a_i^0) &= \max\{I, \val_i\} - I - c_i; \\
                             G^L(I, a_i^1) &= \Ex{\vali} - I, \\
                             G^L(I, \End) &= 0.
                        \end{align*}
        \end{enumerate}
        
    
    \begin{restatable}{claim}{claimLPNOIPNOI}
	  \label{cl:LPNOI2PNOI}
	  Any policy $\alg^L$ for the L-PNOI instance~$B^L$ can be executed on the PNOI instance~$B$, with an expected payoff no less than that of the execution on~$B^L$.
    \end{restatable}
    
	\begin{definition}
	  In the L-PNOI instance~$B^L$, when its state $I^L$ is at least~$\threshold$ and the set of unopened boxes is~$S$, a policy is a \emph{quasi-index policy} if it exhaustively opens (in any order) all boxes $i \in S$ with $\Ex{G^L(I^L, a^0_i)} > 0$, but terminates once its state transits to a higher one.
	\end{definition}


        \begin{restatable}{lemma}{lemmaIndexLPNOI}
            \label{lem:index_LPNOI}
        	  For any $v \in \{\threshold_1, \ldots, \threshold_m \}$, $S \subseteq [n]$, let $W^L(S, v)$ be the expected marginal payoff of any quasi-index policy for~$B^L$.  We have $\PI^L(S, v) \geq (1-\eps) \cdot F(S, v)$.
        \end{restatable}

    \begin{restatable}{lemma}{lemmaHighDiscret}\label{lem:high_discret}  
    Let $\OPT$ and $\OPT^L$ be the optimal expected payoff of the PNOI instance $\Ins$ and the L-PNOI instance $\Ins^L$, respectively.  Then $\OPT \geq \OPT^L \geq (1-O(\epsilon)) \cdot \OPT$.
    \end{restatable}
    
    Finally, given any PNOI instance, one may first discretize its small values, followed by a discretization of the large values, and then apply Theorem~\ref{thm:jian}.  Theorem~\ref{thm:general_ptas} is thus proved.
    The formal argument is given in Appendix~\ref{sec:ptas-app}.
    

\Xcomment{
                Let $\alg$ be an arbitrary policy and $T$ its decision tree in the \Pandora setting.
                Note that $T$ can also be used to describe the execution of $\alg$ in the L-PNOI setting:
                \begin{itemize}
                    \item For each node $u \in T$, add an attribute $I^L_u \coloneqq \UP(I_u)$.
                    \hufu{This seems never used}
                    \item Define $G^L(u) \coloneqq \Ex{g^L(I^L_v, a_v)}$ as the expected marginal payoff at node $u$.
                \end{itemize}
                
                We have
                \begin{align*}
                    \P^L(\alg^*) &= \sum_{u \in T^*} G^L(u) \cdot \Phi(u) \\
                    &= \sum_{\substack{u \in T^* \\ I_u < \threshold}} G(u) \cdot \Phi(u) + \sum_{\substack{u \in T^* \\ I_u \geq \threshold}} G^L(u) \cdot \Phi(u)
                \end{align*}
                \begin{align}
                    \P(\alg^*) - \P^L(\alg^*) &= \sum_{\substack{u \in T^* \\ I_u \geq \threshold}} (G(u) - G^L(u)) \cdot \Phi(u) \\
                    &= \sum_{u \in \SUB} \sum_{w \in T(u)} (G(w) - G^L(w)) \cdot \Phi(w) \\
                    &= \sum_{u \in \SUB} \left(\sum_{\substack{w \in T(u) \\ I_w = I_u}} (G(w) - G^L(w)) \cdot \Phi(w) + \sum_{\substack{w \in T(u) \\ I_w > I_u}} (G(w) - G^L(w)) \cdot \Phi(w) \right)\\
                    &\leq \sum_{u \in \SUB} \left(\sum_{\substack{w \in T(u) \\ I_w = I_u}} (G(w) - G^L(w)) \cdot \Phi(u) + \sum_{\substack{w \in T(u) \\ I_w > I_u}} \OPT \cdot \Phi(w) \right)\\
                    &\leq \sum_{u \in \SUB} \left(\sum_{\substack{w \in T(u) \\ I_w = I_u}} (G(w) - G^L(w)) \cdot \Phi(u) +  \OPT \cdot \eps \cdot \Phi(u) \right)\\ \label{equ:no2jackpot2}
                    &= \sum_{u \in \SUB} \left( (F(T(u), I_u) - F(T(u), \UP(I_u))) \cdot \Phi(u) + \OPT \cdot \eps \cdot \Phi(u) \right)\\
                    &\leq \sum_{u \in \SUB} 2\eps\cdot \OPT \cdot \Phi(u) \\
                    &\leq 2 \eps \cdot \OPT
                \end{align}
            
                Inequality~\eqref{equ:no2jackpot2} is from Lemma~\ref{lem:no2jactpot}. 
}

\section{Conclusion and Open Problems}
\label{sec:conclusion}
In this work, we proved the first computational hardness result for PNOI and improved the state-of-the-art approximation algorithm by a PTAS.
There are other online decision problems where hardness results are missing, and approximation algorithms have been developed in the absence of a tractable optimal algorithm.
The price of information setting for bipartite matching is one such example \citep{singla18poi, gamlath19matching}. 

Our PTAS yields policies for PNOI with arbitrarily good approximations, but its running time has an exponential dependence on $\eps$ which is inherited from \citeauthor{fu2018ptas}'s framework. 
It is an attractive question whether there is an FPTAS for the problem.


\bibliographystyle{plainnat}
\bibliography{bibs.bib}

\appendix

\section{Omitted Proofs from Section~\ref{sec:prelim}}
\label{sec:prelim-app}

\thmindex*

\begin{proof}
Consider any policy for an instance of the classical Pandora box problem.   
We first derive an upper bound on the policy's expected payoff, and then show that the index-based policy achieves the upper bound.

For each $i$, let random variables $I_i$ and $A_i$ be the indicator variables for the events that the policy opens box~$i$ and takes box~$i$, respectively.
As the policy is not allowed to take a sealed box, we have $A_i \leq I_i$ with probability~$1$.
The policy's expected payoff is $\Ex{\sum_i (A_i \vali - I_i c_i)}$.

Note that the policy's decision to open box~$i$ is independent from~$\vali$, therefore $I_i$ and~$\vali$ are independent.  
This allows us to rewrite the expected payoff using the definition of the indices:
\begin{align*}
    \Ex{\sum_i (A_i \vali - I_i c_i)} & = \Ex{\sum_i [A_i \vali - I_i (\vali - \gittins_i)_+]} \\
    & \leq \Ex{\sum_i A_i [\vali - (\vali - \gittins_i)_+]}
    = \Ex{\sum_i A_i \amorval_i} \leq \Ex{\max_i \amorval_i}.
\end{align*}
where in the first inequality we used $A_i \leq I_i$,
in the ensuing equality we used the definition  $\amorval_i = \min \{\vali, \gittins_i\}$, 
and the last inequality follows from the constraint that $\sum_i A_i \leq 1$ with probability~$1$. 

Let us see that the index-based policy's payoff is precisely this upper bound. 
Consider properties of a policy that would turn the two inequalities in the chain into equalities:
\begin{enumerate}
    \item \label{item:non-exposed}
    $I_i(\vali - \gittins_i)_+ = A_i(\vali - \gittins_i)_+$ with probability 1 if, whenever the policy opens a box~$i$ and finds $\vali > \gittins_i$, the policy takes box~$i$;
    \item $\sum_i A_i \amorval_i = \max_i \amorval_i$ with probability 1 if the policy always takes the box with the maximum $\amorval_i$.
\end{enumerate}
It is straightforward to verify that the index-based policy have both properties, and hence attains a payoff equal to the upper bound $\Ex{\max_i \amorval_i}$.
The index-based policy therefore achieves maximum payoff among all policies.
\end{proof}

\subsection{Proof for the Structure Theorem~\ref{thm:structure}}

\thmstructure*

The following terminologies we inherit from \citet{GMS08}.

\begin{definition}
  We say a box is a \emph{backup} in the execution of a policy, if the box is taken without inspection.
\end{definition}

\begin{definition}[$\leV$ tree and $\leV$ path]
For $V \geq 0$, a \emph{$\leV$ tree} is a decision tree that makes the same decisions irrespective of the values of the probed boxes, as long as these values are less than or equal to $V$.
In a $\leV$ tree, decisions constitute a path if all observed values are no larger than $V$; such a path is called a \emph{$\leV$ path}.
\end{definition}


\begin{lemma}
\label{lm:induction}
Suppose an optimal policy $\alg^*$ probes a box $B_j$ at a node $m$ in its decision tree, and suppose when $B_j$ is observed to take value $V$, $\alg^*$ takes a backup box somewhere down the tree.
Then there exists another optimal policy $\alg'$ which has the same decision tree as~$\alg^*$ except possibly for the subtree rooted at $m$.
In $\alg'$, the subtree rooted at $m$ is a $\leV$ tree and takes a backup box at the end of its $\leV$ path.
\end{lemma}
 
\begin{proof}
  The proof proceeds via induction on the value~$V$. 
  
  (Base:) If $V$ is the highest value in the support of all value distributions, the statement is true since it is optimal to end the game immediately and never take a backup (and so the condition cannot be satisfied).
  
(Inductive step:)  Now suppose the statement is true for values larger than~$V$. 
Let $u$ be the child node of~$m$ where box~$B_j$ yields value~$V$, then the subtree rooted at~$u$ is without loss of generality a $\leV$ tree.
Suppose somewhere in this subtree, a backup is taken, we first show that it is without loss of generality to assume that the end of the $\leV$ path from $u$ is a backup.
Suppose this is not the case, 
then some node along the $\leV$ path from~$u$ has a child not on this path which has a backup descendant; let $w$ be such a node (on the $\leV$ path) closest to~$u$.  
($w$ may be $u$ itself.) 
Since the said child of~$w$ with a backup descendant is not on the $\leV$ path, the box opened at~$w$ must yield a value $V' > V$ to arrive at that child.
By the induction hypothesis, we may assume that the subtree rooted at~$w$ is a $\leq_{V'}$ tree, and the $\leq_{V'}$ path originating from~$w$ ends in a backup.
But the $\leq_{V'}$ path from~$w$ is part of the $\leV$ path from~$u$, and therefore the end of the $\leV$ path from $u$ is a backup.

For each value~$h$ possibly taken by box~$B_j$, let $T_h$ be the subtree rooted at the child of~$w$ where $v_{B_j} = h$.
(So $T_V$ is the subtree rooted at~$u$.)
The crucial observation is that, for every $h < V$, one may replace the subtree $T_h$ by $T_{V}$ without decreasing the expected payoff.
(To see this, note first that box~$B_j$ is never used in $T_{V}$: it is not used along the $\leV$ path because a backup is used in the end; $B_j$ is not used elsewhere either, because one must see a value larger than $V$ to leave the $\leV$ path in the first place, and that value is preferred to $v_{B_j} = V$.  Therefore, it is feasible to replace $T_h$ by $T_{V}$.  Such a replacement does not decrease the expected payoff, because the expected payoff of $T_h$ is no more than that of $T_{V}$ for $h < V$, by the optimality of~$\alg^*$.)
After such replacements, the subtree rooted at~$m$ becomes a $\leV$ tree, and its endpoint is a backup, the same backup as the end of the $\leV$ path from~$u$.
The resulting policy is the $\alg^*$ stated by the lemma.



\end{proof}

\begin{proof}[Proof of Theorem~\ref{thm:struc}]
  Let $\alg^*$ be an optimal policy.
Consider the non-trivial case when $\alg^*$ probes at least one box and takes a backup box somewhere.
Suppose $\alg^*$ probes box~$i$ first.
 Let $V_i$ be the maximum value taken by box~$i$ such that a backup is possibly taken later.  
Then, applying Lemma~\ref{lm:induction} to the root of the decision tree, we may assume the decision tree is a $\leq_{V_i}$ tree. 
If box~$i$ yields a value larger than $V_i$, then by the maximality of~$V_i$, no backup is used downstream, and (without loss of generality) $\alg^*$ follows the index policy on the remaining boxes, taking $v_i$ as a free outside option.
Now apply the same argument to the next node along the $\leq_{V_i}$ path.  The theorem follows by a repeated application of this argument until the next node in line is a backup: as long as the next node probes a box, the threshold of that box is given by the largest value it can take to still see some backup somewhere downstream; the ordering $\sigma$ is given by the order in which boxes are probed by the nodes on which Lemma~\ref{lm:induction} is applied.  
\end{proof}

\section{Appendix for Section~\ref{sec:hardness}}
\label{sec:hardness-app}
\supporttwo*
\begin{proof}
A few observations are in order.  
\begin{enumerate}[(1)]
\item A box~$i$ to be taken without inspection can be seen as a box with deterministic value~$\Ex{\vali}$ with no search cost.
Therefore, by the optimality of the index-based policy in the classic Pandora box problem,
one should never take a box~$i$ without inspection if some other unopened box has index larger than~$\Ex{\vali}$.

\item  When a box yields value~$1$ upon inspection, it is payoff-optimal to select the box immediately and quit.  

\item For any policy satisfying (2), the way in which it opens boxes is completely described by a subset $S \subseteq [n]$ and a permutation~$\pi$ on~$S$.
The policy opens boxes in~$S$ in the order specified by~$\pi$: if a box yields value~$1$, the box is taken immediately and the search terminates; otherwise, this goes till all boxes in~$S$ are opened and yield value~$0$, at which point the algorithm may terminate or take a box not in~$S$ without inspection.

\end{enumerate}
From these observations, it is without loss of generality to consider policies that: 
(i) commit to a certain box~$i$ that is possibly selected without inspection;
(ii) inspect boxes with indices at least~$\Ex{\vali}$ in decreasing order of their indices, and if a value~$1$ is found, select that box and quit;
(iii) when all boxes in step~(ii) yield value~$0$, take box~$i$ without inspection and quit.

There are altogether $n$ such policies (up to tie-breaking in step~(ii), which does not matter). 
We can enumerate them and choose the best one in polynomial time.

\end{proof}
\normal*
\begin{proof}
We prove the three properties of a normal policy in order. 

\begin{itemize}
\item It is straightforward to see that $\alg$ should take the last box without inspection if all previous boxes yield value~0.
To see that this is the only situation $\alg$ should bypass inspection, recall by observation~(1) in the proof of Proposition~\ref{prop:support2} that an optimal policy should not take a box without inspection if there are other unopened boxes with higher indexes. 
Note that, for any two boxes $i$ and~$j$, $\Ex{\vali} < \frac{1}{2} \leq \gittins_j$ by definition of LCLRS3 instances.

\item It is straightforward that $\alg$ stops when it sees value~$1$ --- no other box can yield higher values in an LCLRS3 instance, and opening more boxes strictly diminishes the payoff.

\item Since any unopened box has expected value strictly smaller than~$\tfrac 1 2$, $\alg$ should never take an unopened box without inspection if a value~$\tfrac 1 2$ is already seen.
In other words, with value $\tfrac 1 2$ seen, $\alg$ ignores the option to bypass inspection, and the problem degenerates to the classical Pandora box problem for the remaining boxes.
Therefore, after a value~$\tfrac 1 2$ is seen, $\alg$ runs the index-based policy on the remaining boxes.
\end{itemize}
\end{proof}

\utilitypermutation*
\begin{proof}
We first modify $\alg$ to obtain another policy~$\alg'$. 
$\alg'$ has the same behavior as~$\alg$ in all situations, except that when $\alg$ takes the last box $\sigma(n)$ without opening it, $\alg'$ opens box~$\sigma(n)$, pays the cost~$c_{\sigma(n)}$, and takes it.
Since $\alg$ is normal, it takes box~$\sigma(n)$ only when all the other boxes have yielded value~$0$, which happens with probability $\prod_{i=1}^{n-1} r_{\sigma(i)}$. 
So we have 
\begin{align}
\label{eq:a-a'-diff}
    \P(\alg) = \P(\alg') + c_{\sigma(n)}\prod_{i = 1}^{n-1} r_{\sigma(i)}.
    \end{align}

Next, we derive a simple expression for $\P(\alg')$ and then show that the difference between $\P(\alg')$ and $\P(\alg_P)$ gives rise to the second term in~\eqref{eq:diff-pandora}.
We claim $\P(\alg') = \sum_i \Ex{A'_i \amorval_i}$, where $A_i'$ is the indicator variable for the event that $\alg'$ takes box~$i$.
To see this, we make two observations, which amounts to showing that $\alg'$ is non-exposed (Definition~\ref{def:non-exposed}).
Let $I'_i$ be the indicator variable for the event that $\alg'$ opens box~$i$.  

  \begin{itemize}
    \item $\alg'$ never exercises the option to take a box without opening it, so $A'_i \leq I'_i$ with probability~$1$, for all~$i$.
    \item Whenever $\alg'$ opens box~$i$ and sees $\vali > \gittins_i$, $\alg'$ immediately takes box~$i$.
    To see this, if $i = \sigma(n)$, by definition $\alg'$ takes the box after opening it.
    For $i \neq \sigma(n)$, by definition of LCLRS3, $\gittins_i \geq \frac 1 2$, so $\vali > \gittins_i$ implies $\vali = 1$.
    Since $\alg$ is normal, it immediately takes box~$i$ when seeing $\vali = 1$; $\alg'$ copies the behavior of~$\alg$ for $i \neq \sigma(n)$, and hence also immediately takes it.
\end{itemize}




From the second observation, we have $(\vali - \gittins_i)_+ I'_i = (\vali - \gittins)_+ A'_i$ with probability~$1$.
Therefore
  \begin{align*}
    \P(\alg') & = \sum_i \Ex{\vali A'_i - c_i I'_i} = \sum_i \Ex{\vali A'_i - (\vali - \gittins_i)_+ I'_i} \\
    & = \sum_i \Ex{\vali A'_i - (\vali - \gittins_i)_+ A'_i} = \sum_i \Ex{A'_i \amorval_i}.
  \end{align*}
  
  We now compare $\P(\alg_P)$ and $\P(\alg')$.  
  By Theorem~\ref{thm:index},  $\P(\alg_P)=\Ex{\max_i \amorval_i}$.
  For every realization of $\amorval_i$'s, $\max_i \amorval_i \geq \sum_i A'_i \amorval_i$.  
  The inequality is strict only when $\alg'$ takes some box~$i$ with $\amorval_i < \max_j \amorval_j$.
  This can happen only when all boxes opened before~$i$ yield value~$0$, and $\vali = 1$, in which case $\amorval_i = \gittins_i$.
  This happens with probability $p_i \prod_{j = 1}^{\sigma^{-1}(i) - 1} r_{\sigma(j)}$;
  the expected contribution to the difference between $\max_i \amorval_i$ and $\gittins_i$ conditioning on this happening is $\Ex{(\max_{j: \sigma^{-1}(j) > \sigma^{-1}(i)} \gittins_j - \gittins_i)_+}$, which is just $g(i, T_{\sigma}(i))$ as defined in the statement of the lemma.
  Therefore, overall, we have
  \begin{align}
      \P(\alg_P) - \P(\alg') &= \Ex{\max_i \amorval_i} - \sum_i \Ex{A'_i \amorval_i}  \nonumber
      \\
      &= \sum_i \Pr[A'_i = 1] \cdot \Ex{\max_j \amorval_j - \amorval_i \mid A'_i = 1} \nonumber \\
      &= \sum_i p_{i}g(i, T_\sigma(i)) \prod_{j = 1}^{\sigma^{-1}(i)-1} r_{\sigma(j)}.
  \label{eq:ap-a'-diff}
  \end{align}
  
Combining \eqref{eq:a-a'-diff} and \eqref{eq:ap-a'-diff}, we have
\begin{align*}
    \P(\alg_P)=\P(\alg)+ \sum_i p_{i}g(i, T_\sigma(i)) \prod_{j = 1}^{\sigma^{-1}(i)-1} r_{\sigma(j)} 
      - c_{\sigma(n)}\prod_{i = 1}^{n-1} r_{\sigma(i)}.
\end{align*}
\end{proof}

\claimFixFinalBox*
\begin{proof}
  By \eqref{eq:object}, it is easy to verify that when $n\ge 1$ and $\sigma^*(n+2)={n+2}$ one has
  \begin{align*}
        & ~\Ut(\sigma^*)\\
        \geq& ~c_{n+2}r_{n+1}(1-2^{-6n})^{n}-np_{n+2}\max_{i}p_i(\tau_H-\tau_i) - p_{n+1}\Big((\tau_H-\tau_L)p_{n+2}+n\max_{i\in[n]}p_i(\tau_i-\tau_L)\Big) \\
        \geq& ~ \frac{40}{32\Gamma}(1-2^{-6n})^n- \frac{1}{32\Gamma} - O(n\Delta^2)\\
        \geq& ~\frac{38}{32\Gamma}.
  \end{align*}
  
For any permutation $\sigma$ with $\sigma(n+2)\in[n]$, we have
  \begin{align*}
      \Ut(\sigma)\leq  r_{n+1}\max_{i\in[n]}c_i \leq O(\Delta^2)
  \end{align*}
  as $c_i=\frac{p_i}{\tau_i+1}\leq p_i$.

For any permutation $\sigma$ with $\sigma(n+2)=n+1$, we have
\begin{align*}
    \Ut(\sigma)\leq r_{n+2}c_{n+1}
    \leq \frac{3}{8\Gamma}<\frac{38}{32\Gamma}.
\end{align*}
Hence if some permutation $\sigma^*$ maximize Equation~\eqref{eq:object}, then $\sigma^*(n+2)=n+2$.
\end{proof}

\subsection{Omitted Proof from Appendix~\ref{sec:reduction-correct}}
\label{sec:proof_for_sketch}
For ease of presentation, we introduce the following notations.
\begin{definition}
    \begin{align*}
        g_{H} & \coloneqq p_{n+2}(\tau_{n+2}-\tau_H); \\
        g_L & \coloneqq \left[1-\prod_{i\in \Tat}(1-p_i)\right](1-p_{n+2})(\tau_{H}-\tau_L)+p_{n+2}(\tau_{n+2}-\tau_L); \\
        g_i & \coloneqq g(i, T_{\sigma}(i)), \quad \textnormal{for } i = 1, \ldots, n + 1.\\
    \end{align*}
\end{definition}

\begin{restatable}{claim}{claimGi}
    \label{clm:g_ig_n+1}
        \begin{align*}
        g_i & = g_{H} - \frac{p_i p_{n+1} (1-p_{n+2})(\tau_H - \tau_L)} 2 \pm O(n\Delta^3), \quad \textnormal{for } i = 1, \ldots, n; \\
        g_{n+1} & = g_{L} \pm O(n\Delta^3).
        \end{align*}
\end{restatable}
    \begin{proof}
        Recall that $g(i, T) \coloneqq \Ex{(\max_{j \in T} \amorval_j - \gittins_i)_+}$ and $\amorval_i = \min \{\vali, \gittins_i\}$. Also, $\tau_{n+2}>\tau_i > \gittins_H >\tau_L = 1/2$ for all $i \in [n]$ in our LCLRS3 instance.
        Therefore, for each $j \neq i$,
        \begin{displaymath}
            (\amorval_j - \gittins_i)_+ = \left\{
            \begin{array}{ll}
                0, &\textrm{ if $\val_j \leq 1/2$;} \\
                (\gittins_j - \gittins_i)_+, &\textrm{ if $\val_j = 1$.}
            \end{array}
            \right.
        \end{displaymath}
        Since $\gittins_{n+2}$ is by far the largest among all indices, and $\gittins_{n+1} = \gittins_L$ is the lowest index,
        we have for each $i \in [n]$, 
        \begin{align*}
            g_i &= p_{n+2}(\tau_{n+2}-\tau_i) + (1-p_{n+2})\Ex{\max_{j \in T_{\sigma}(i) \backslash \{n+2\}} (\amorval_j - \gittins_i)_+} \\
            &= p_{n+2}(\tau_{n+2}-\tau_H) - \frac 1 2 p_i p_{n+1} (1-p_{n+2})(\tau_H - \tau_L) + (1-p_{n+2})\Ex{\max_{j \in T_{\sigma}(i) \backslash \{n+2\}} (\amorval_j - \gittins_i)_+} \\
            &\leq g_{H} - \frac 1 2 p_i p_{n+1} (1-p_{n+2})(\tau_H - \tau_L) + (1-p_{n+2})\sum_{j \in T_{\sigma}(i) \backslash \{n+2\}} p_j (\gittins_j - \gittins_i)_+ \\
            &\leq g_{H} - \frac 1 2  p_i p_{n+1} (1-p_{n+2})(\tau_H - \tau_L) \pm O(n\Delta^3).
        \end{align*}
        The last inequality is from the fact that $\tau_i = \tau_H + O(\Delta^2)$ for any $i \in [n]$.
        
        Similarly,
        \begin{align*}
            g_{n+1} &= p_{n+2}(\tau_{n+2}-\tau_L) + (1-p_{n+2})\Ex{\max_{j \in \Tat} (\amorval_j - \tau_L)_+} \\
            &\leq p_{n+2}(\tau_{n+2}-\tau_L) + (1-p_{n+2})(1-\Pi_{i\in \Tat}(1-p_i)){\max_{j \in \Tat} (\tau_j - \tau_L)} \\
            &\leq p_{n+2}(\tau_{n+2}-\tau_L) + (1-p_{n+2})(1-\Pi_{i\in \Tat}(1-p_i))(\tau_H \pm O(\Delta^2)- \tau_L)\\
            &= g_{L} \pm O(n\Delta^3).
        \end{align*}
    \end{proof}

 \lemmaFromScheduling*
    \begin{proof}
We prove this lemma by induction on $n$.
When $n=1$, we get $p_1r_0=p_1=c(1-r_1)$ by the assumption.

Suppose one can have $\sum_{i=1}^{n}p_i\prod_{j=0}^{i-1}r_j=c(1-\prod_{i=1}^{n}r_i)$,
then
\begin{align*}
    \sum_{i=1}^{n+1}p_i\prod_{j=0}^{i-1}r_j=& \sum_{i=1}^{n}p_i\prod_{j=0}^{i-1}r_j+p_{n+1}\prod_{j=0}^nr_j\\
    =& c(1-\prod_{i=1}^{n}r_i)+p_{n+1}\prod_{j=0}^nr_j\\
    =& c(1-\prod_{i=1}^{n}r_i)+c(1-r_{n+1})\prod_{j=0}^nr_j\\
    =& c(1-\prod_{i=1}^{n+1}r_i).
\end{align*}
\end{proof}

\claimBeforeApprox*

\begin{proof}
First, we show
    \begin{align*}
        \Ls(\sigma)=&\frac {g_H}2 \left(1-\prod_{i\in \Tbf}r_i \right)  + g_L p_{n+1}\prod_{i\in\Tbf}r_i+\frac{g_Hr_{n+1}} 2\prod_{i\in \Tbf}r_i \left(1-\prod_{i\in\Tat}r_i \right)\\
        &~+ \frac 1 2 \sum_{i\in \Tbf} p_{i}^2 p_{n+1}(1-p_{n+2})(\tau_H - \tau_L) \pm O(n^2\Delta^4).
    \end{align*}
    The claim follows by plugging in the parameters.
One has
        \begin{align*}
        \Ls(\sigma) 
        =& \sum_i p_{i} g(i, T_\sigma(i)) \prod_{j = 1}^{\sigma^{-1}(i)-1} r_{\sigma(j)} \\
        =& \sum_{i\in \Tbf} p_{i} g_i \prod_{j = 1}^{\sigma^{-1}(i)-1} r_{\sigma(j)} + p_{n+1}g_{n+1}\prod_{i\in\Tbf}r_i+\sum_{i\in\Tat}p_i g_i \prod_{j = 1}^{\sigma^{-1}(i)-1} r_{\sigma(j)} \\
        =& \sum_{i\in \Tbf} p_{i} \left(g_H - \frac 1 2 p_i p_{n+1} (1-p_{n+2})(\tau_H - \tau_L) \pm O(n\Delta^3)\right) \prod_{j = 1}^{\sigma^{-1}(i)-1} r_{\sigma(j)} \\
        &+ p_{n+1} (g_L \pm O(n\Delta^3)) \prod_{i\in\Tbf}r_i +r_{n+1}\prod_{s\in \Tbf}r_s\sum_{i \in \Tat}p_i g_i \prod_{j=|\Tbf+2|}^{\sigma^{-1}(i)-1}r_{\sigma(j)}\\
        =& \sum_{i\in \Tbf} p_{i} g_H \prod_{j = 1}^{\sigma^{-1}(i)-1} r_{\sigma(j)} - \sum_{i\in \Tbf} \left(\frac 1 2 p_{i}^2 p_{n+1}(1-p_{n+2})(\tau_H - \tau_L) \right) \prod_{j = 1}^{\sigma^{-1}(i)-1}r_{\sigma(j)}\\
        &+
        p_{n+1} g_L \prod_{i\in\Tbf}r_i+r_{n+1}\prod_{s\in \Tbf}r_s\sum_{i\in \Tat} \prod_{j=|\Tbf+2|}^{\sigma^{-1}(i)-1}r_{\sigma(j)}p_ig_i  \pm O(n^2\Delta^4)\\
        =& \sum_{i\in \Tbf} p_{i} g_H \prod_{j = 1}^{\sigma^{-1}(i)-1} r_{\sigma(j)} - \sum_{i\in \Tbf} \left(\frac {1} {2} p_{i}^2 p_{n+1}(1-p_{n+2})(\tau_H - \tau_L)\right) \prod_{j = 1}^{\sigma^{-1}(i)-1}r_{\sigma(j)}\\
        & +r_{n+1}\prod_{s\in \Tbf}r_s\sum_{i\in \Tat} \prod_{j=|\Tbf+2|}^{\sigma^{-1}(i)-1}r_{\sigma(j)}p_i\Big(g_{H} - \frac{p_i p_{n+1} (1-p_{n+2})(\tau_H - \tau_L)} 2 \pm O (n\Delta^3) \Big)\\
        & +  p_{n+1} g_L \prod_{i\in\Tbf}r_i\pm O(n^2\Delta^4)\\
        =& \sum_{i\in \Tbf} p_{i} g_H \prod_{j = 1}^{\sigma^{-1}(i)-1} r_{\sigma(j)} - \sum_{i\in \Tbf} \left(\frac 1 2 p_{i}^2 p_{n+1}(1-p_{n+2})(\tau_H - \tau_L)\right) \prod_{j = 1}^{\sigma^{-1}(i)-1}r_{\sigma(j)}\\
        & +r_{n+1}\prod_{s\in \Tbf}r_s\sum_{i\in \Tat} \prod_{j=|\Tbf+2|}^{\sigma^{-1}(i)-1}r_{\sigma(j)}p_i\left(g_{H} - \frac 1 2 {p_i p_{n+1} (1-p_{n+2})(\tau_H - \tau_L)} \right)\\
        & +  p_{n+1} g_L \prod_{i\in\Tbf}r_i\pm O(n^2\Delta^4),
        \end{align*}
where in the last equality we used $p_{n+1} = 1 / \Gamma = O(\Delta)$.  
    Further analyzing the last term, we have
    \begin{align*}
        & r_{n+1}\prod_{s\in \Tbf}r_s\sum_{i\in \Tat} \prod_{j=|\Tbf+2|}^{\sigma^{-1}(i)-1}r_{\sigma(j)}p_i \left(g_{H} - \frac 12 {p_i p_{n+1} (1-p_{n+2})(\tau_H - \tau_L)} \right)\\
        & =r_{n+1} \prod_{s\in \Tbf}r_s\sum_{i\in \Tat} \prod_{j=|\Tbf+2|}^{\sigma^{-1}(i)-1}r_{\sigma(j)}p_ig_H-O(n\Delta^4).
    \end{align*}
        Note that $ \prod_{j = 1}^{\sigma^{-1}(i)-1}r_{\sigma(j)}= 1 - O(n\Delta)$ for each $i\in\Tbf$ and $r_{n+1} = 3/\Gamma=O(\Delta)$.
        Besides, notice that, for each $i \in [n]$, we have $\frac{p_i}{1 - r_i} = \frac{s_i / \Gamma}{2s_i / \Gamma} = \frac 1 2$.
        Therefore, Lemma~\ref{lm:from_scheduling} applies, and we have
        \begin{align*}
        \Ls(\sigma)=& \frac{g_H} 2\left(1-\prod_{i\in\Tbf}r_i \right) + (1-O(n\Delta))\sum_{i\in \Tbf} \frac 1 2 p_{i}^2 p_{n+1}(1-p_{n+2})(\tau_H - \tau_L)  + p_{n+1} g_L \prod_{i\in\Tbf}r_i \\
        &~~+r_{n+1}\prod_{s\in \Tbf}r_s\sum_{i\in\Tat} \prod_{j=|\Tbf+2|}^{\sigma^{-1}(i)-1}r_{\sigma(j)}p_ig_H \pm O(n^2\Delta^4)
        \\
        =& \frac{g_H} 2\left(1-\prod_{i\in\Tbf}r_i \right) + \sum_{i\in \Tbf} \frac 1 2 p_{i}^2 p_{n+1}(1-p_{n+2})(\tau_H - \tau_L)  + p_{n+1} g_L \prod_{i\in\Tbf}r_i\\
        &~~+ \frac 1 2 g_Hr_{n+1}\prod_{s\in \Tbf}r_s \left(1-\prod_{i\in \Tat}r_i \right) \pm O(n^2\Delta^4).
        \end{align*}
        
        
        Rewrite the equation in terms of $f, g$ and $k_1, k_2$, and we have
        \begin{align*}
        \Ls(\sigma)=&  \left(-\frac{g_H}{2}+g_Lp_{n+1}+\frac{g_Hr_{n+1}}{2} \right)\prod_{i\in\Tbf}r_i+\frac{1}{2}p_{n+1}(1-p_{n+2})(\tau_H - \tau_L) \sum_{i\in \Tbf}p_i^2\\
        &+ \frac{g_H}{2}-\frac{1}{2}g_H\prod_{i\in [n+1]}r_i \pm O(n^2\Delta^4)\\
        = & \Big(-\frac{g_H(p_{n+1}+q_{n+1})}{2}+p_{n+1}( (1-p_{n+2})(\tau_{H}-\tau_L)+p_{n+2}(\tau_{n+2}-\tau_L) )\Big)\prod_{i\in\Tbf}r_i\\
        & - p_{n+1} (1-p_{n+2}) (\tau_H-\tau_L) \prod_{i\in\Tbf} r_i \prod_{j\in \Tat} (1-p_j) + \frac{1}{2}p_{n+1}(1-p_{n+2})(\tau_H - \tau_L) \sum_{i\in \Tbf}p_i^2\\
        &+ \frac{p_{n+2} (\tau_{n+2}-\tau_H)}{2} \left(1-\prod_{i\in[n+1]}r_i \right) \pm O(n^2\Delta^4)\\
        =& k_1 f(\Tbf)-k_2 f(\Tbf)g(\Tat)- \frac 1 2 k_2\sum_{i\in \Tat} p_{i}^2 + C \pm O(n^2\Delta^4).
        \end{align*}
    \end{proof}
    
    \claimApproxError*
    \begin{proof}
Consider $f(T)$ first.
By Fact~\ref{fct:Taylor_series}, as $p_i\leq 2^n/\Gamma=\Delta$, we know $ e^{-2(p_i+p_i^2)}=1-2p_i+O(p_i^3)$.
Then $f(T)=\Pi_{i\in T}(1-2p_i)=\Pi_{i\in T}(e^{-2(p_i+p_i^2)}-O(p_i^3))\geq \Pi_{i\in T}e^{-2(p_i+p_i^2)}-{|T|}\max_iO(p_i^3)$.

As for the other side, again by Fact~\ref{fct:Taylor_series}
we have $1-2p_i\leq e^{-2(p_i+p_i^2)}$, which directly implies $e^{-\sum_{i\in T}2(p_i+p_i^2)}\geq f(T) $.

The other inequality is based on $ e^{-(p_i+p_i^2/2)}-O(p_i^3) \le 1-p_i\leq e^{-(p_i+p_i^2/2)} $, therefore it holds by the same argument.
\end{proof}

\claimSettingt*

\begin{proof}
We would like 
\begin{align*}
k_2=tk_1=tk_2+t[-p_{n+2}(\tau_{n+2}-\tau_H)(p_{n+1}+q_{n+1})/2+p_{n+1}p_{n+2}(\tau_{n+2}-\tau_{L})],    
\end{align*}
which is equivalent to
\begin{align*}
    (1-t)p_{n+1}(1-p_{n+2})(\tau_{H}-\tau_L)=t[-p_{n+2}(\tau_{n+2}-\tau_H)(p_{n+1}+q_{n+1})/2+p_{n+1}p_{n+2}(\tau_{n+2}-\tau_{L})].
\end{align*}
Rearranging the terms, we get
\begin{align*}
    & [(1-t)p_{n+1}(1-p_{n+2})-tp_{n+2}(p_{n+1}+q_{n+1})/2]\tau_H\\
    =& t[-p_{n+2}\tau_{n+2}(p_{n+1}+q_{n+1})/2+p_{n+1}p_{n+2}(\tau_{n+2}-\tau_L)]+(1-t)p_{n+1}(1-p_{n+2})\tau_L.
\end{align*}
Recall that $p_{n+2}=1/8,\tau_{n+2}=3/4,p_{n+1}=1/\Gamma,q_{n+1}=1-41/\Gamma,\tau_L=1/2$ and thus
\begin{align*}
    &\left[\frac{7(1-t)}{8\Gamma}-\frac{t}{16}\left(1-\frac{40}{\Gamma}\right) \right]\tau_H=\frac{7(1-t)}{16\Gamma}-\frac{3t}{64}\left(1-\frac{40}{\Gamma} \right)+\frac{t}{32\Gamma}\\
    \Rightarrow &~~\tau_H =\frac{\frac{7(1-t)}{16\Gamma}-\frac{3t}{64}(1-\frac{40}{\Gamma})+\frac{t}{32\Gamma}}{\frac{7(1-t)}{8\Gamma}-\frac{t}{16}(1-\frac{40}{\Gamma})}=\frac{-3t\Gamma+28+94t}{-4t\Gamma+56+104t}.
\end{align*}

We need $1/2<\tau_H<3/4$, which means 
\begin{align*}
    \frac{1}{2}<\tau_H=\frac{-3t\Gamma+28+94t}{-4t\Gamma+56+104t}<\frac{3}{4}.
\end{align*}
For the first inequality, noting that $|t- 2|\leq O(y+\Delta^2)=O(\Delta^2)$ and $\Delta=2^{-7n}$, we have
\begin{align*}
    -4t\Gamma+56+104t>-6t\Gamma+2(28+94t),
\end{align*}
which means the first inequality holds.

As for the second inequality, we need
\begin{align*}
    -12t\Gamma+4(28+94t)>-12t\Gamma+3(56+104)t, \\
    \iff 64t>56,
\end{align*}
which holds by our choice of $t$.
\end{proof}

\lemmaApproxLossByh*

\begin{proof} 
By the setup of the parameters, we have $k_1,k_2 > 0$ and $|k_2/k_1-2e^{y/2}|\leq\Delta^2$.
Combining Claim~\ref{lm:before_approx} and Claim~\ref{clm:approx_error} gives
\begin{align*}
    &~ k_1 \left(e^{-\sum_{i\in \Tbf}2(p_i+p_i^2)} \right) \left(1-\frac{k_2}{k_1}e^{-\sum_{i\in\Tat}(p_i+p_i^2/2)} \right)\\
    \leq &~ \Ls(\sigma)-C+ \frac 1 2 k_2\sum_{i\in \Tat} p_{i}^2  \pm O(n^2\Delta^4),\\
    \leq &~ k_1 \left(e^{-\sum_{i\in \Tbf}2(p_i+p_i^2)}-O(n\Delta^3) \right) \left(1-\frac{k_2}{k_1} \left(e^{-\sum_{i\in\Tat}(p_i+p_i^2/2)}-O(n\Delta^3) \right) \right)\\
    \leq & ~k_1 \left(e^{-\sum_{i\in \Tbf}2(p_i+p_i^2)} \right) \left(1-\frac{k_2}{k_1}e^{-\sum_{i\in\Tat}(p_i+p_i^2/2)} \right) + O(|k_1| n\Delta^3).
\end{align*}

Recall that
 $y:=\sum_{i\in S}s_i/\Gamma+(s_i/\Gamma)^2=\sum_{i\in\Tat\cup\Tbf}p_i+p_i^2$, $x:=\sum_{i\in \Tbf}p_i+p_i^2$ and we define $z:=\sum_{i\in \Tat}p_i^2/2$ for analysis.
We know that $e^{-\sum_{i\in \Tbf}2(p_i+p_i^2)}=e^{-2x}$ and $e^{-y+x+z}= e^{-\sum_{i\in\Tat}(p_i+p_i^2/2)}$.
Note that $0<k_1=O(1)$ and we can rewrite the equations above as
\begin{align*}
    &e^{-2x}\left(1-\frac{k_2}{k_1}e^{-y+x+z} \right)\nonumber \\
    \leq &\frac{\Ls(\sigma)-C+ k_2\sum_{i\in \Tat} p_{i}^2 / 2 \pm O(n^2\Delta^4)}{k_1}\nonumber \\
    \leq & e^{-2x}\left(1-\frac{k_2}{k_1} e^{-y+x+z} \right)+ O(n\Delta^3).
\end{align*}

Note that $1+z\le e^{z}\le 1+z+z^2$ as $0\le z\le 1/4$ and $k_2/k_1\approx 2$.
On one hand, we know that
\begin{align*}
    e^{-2x}\left(1-\frac{k_2}{k_1}e^{-y+x+z} \right)\leq & e^{-2x} \left(1-\frac{k_2}{k_1}e^{-y+x}(1+z) \right).
    \\
    \leq & e^{-2x} \left(1-\frac{k_2}{k_1}e^{-y+x} \right)-\frac{k_2}{k_1}z+O(xz).
\end{align*}
On the other hand,
\begin{align*}
    e^{-2x} \left(1-\frac{k_2}{k_1}e^{-y+x+z} \right)\geq & e^{-2x} \left(1-\frac{k_2}{k_1}e^{-y+x}(1+z+z^2) \right)\\
    \geq & e^{-2x} \left(1-\frac{k_2}{k_1}e^{-y+x} \right)-\frac{k_2}{k_1}z -O(z^2).
\end{align*}

Noting that $O(xz)=O(n\Delta^3)$ and $O(z^2)\le O(n^2\Delta^4)$ completes the proof.

\end{proof}

\claimStrongConvexity*

\begin{proof}
    Recall $h(x)=e^{-2x}(1-\frac{k_2}{k_1}e^{-y+x})$. 
    We have $\frac{\dd h(x)}{\dd x}=-2e^{-2x}+\frac{k_2}{k_1}e^{-y-x}$, and
    $\frac{\dd^2 h(x)}{\dd x^2}=4e^{-2x}-\frac{k_2}{k_1}(e^{-y-x})\in [1,4]$ for $-2^{-6n}\leq x\leq 1/2$. Therefore, $\frac{\dd h(x)}{\dd x}\mid_{x^*}=0$. 
    Hence by strong convexity, we know
    $h(x^*+\epsilon)\geq h(x^*)+\epsilon^2/2$.
    
    Now we prove $|x^*-y/2|\leq O(\Delta^2)$.
We know the derivative $|h'(y/2)|=|-2e^{-y}+\frac{k_2}{k_1}e^{-3y/2}|\leq O(\Delta^2)$, which means $|x^* - y/2|\leq O(\Delta^2)$ by the strong convexity.
\end{proof}

\section{Omitted Proofs from Section~\ref{sec:ptas}}
\label{sec:ptas-app}

\subsection{Omitted Proofs from Section~\ref{sec:SSDP}}\label{sec:ptas:app_tree}
    
    \theoremJian*
    
    \begin{proof}[Proof Sketch]
        A policy is described by a decision tree, where each node corresponds to an action to be taken (see Section~\ref{sec:prelim}).
    It can be shown that when $|V| = O(1)$,  every policy may be approximated, with a loss of at most $O(\eps) \MAX$, by an ensemble of policies that is \emph{block adaptive}.
    A block adaptive policy corresponds to a decision tree whose nodes may be grouped into a small number of \emph{blocks}; within each block, the order of actions affects the expected payoff negligibly, and so the blocks may be seen as supernodes.
    When one condenses the nodes to such supernodes, the tree's depth and the maximum degree of each node are both bounded by constants; one may therefore enumerate the topologies of all such trees.
   For each topology, there are exponentially many ways to fill actions into each block, but each block may be approximated by a \emph{signature vector} (signifying the state in a block and the discretized transition probabilities to the other blocks); there are only polynomially many possible signatures, so they can be efficiently enumerated as well. 
    Finally, a dynamic programming is employed to check whether a given set of signatures, one on each node, in a given tree topology can be realized with actual actions.  
    \end{proof}
    
    \paragraph{Decision Tree}
    Given a policy~$\alg$, for each node $u$ in the decision tree, let $I_u$ be its state, let $a_u$ be the action to be taken by the policy, let $S_u$ be the set of boxes that haven't been opened yet,
        let $T(u)$ be the subtree rooted at~$u$,
        let $G(u) = \Ex{G(I_u, a_u)}$ be the expected marginal payoff at node~$u$, 
        let $\Phi(u)$ be the probability with which $u$ is reached in an execution of the policy, 
        and let $H(u)$ be the sum of marginal payoffs accumulated when the process reaches~$u$.
            The expected payoff of~$\alg$ can then be written in two ways:
            \begin{equation}\label{equ:payoff_by_marginal}
                \P(\alg) = \sum_{u} G(u) \cdot \Phi(u) = \sum_{u \textrm{ is leaf}} H(u) \cdot \Phi(u) \; .
            \end{equation}
            

    \lemmaBad*
    
    \begin{proof}
            It suffices to prove that, in the decision tree $T^*$ of an optimal policy $\alg^*$, for any node $u$, $G(u) \geq 0$.
            Assume, towards a contradiction, that $G(u) < 0$ for a node~$u$.
            $a_u$ cannot be in $A^1$, as an optimal policy cannot take without inspection a box with an expected value smaller than the maximum value seen so far.
            If $a_u = a^0_i \in \ActSet^0$ for some $i \in [n]$, we construct a modified policy~$\alg'$ and show it strictly outperforms~$\alg^*$.
            $\alg'$ uses the same decision tree $T^*$, 
            with the only difference being that when $\alg'$ reaches node $u$, 
            instead of probing box~$i$ as $\alg^*$ does, $\alg'$ samples a value $\val'_i$ from $F_i$
            and simulates the rest of $\alg^*$ \emph{pretending} that $\vali$, which is not observed, is equal to~$\val'_i$.
            $\alg^*$ and $\alg'$ reach every subsequent node~$w$ with the same probability. 
            If $\alg^*$ quits by calling $\End$ and taking box~$i$, $\alg'$ instead takes the maximum value seen so far.\footnote{If there is no other opened box, $\alg'$ may quit without taking anything, or take any box without inspection.}
            
            For each subsequent node~$w$, let $I'_w$ be the minimum between $I_w$ and (the actual) $\vali$.  
            $I'_w$ may be seen as the \emph{true} state at~$w$ in the execution of~$\alg'$.
            Then $I'_w \leq I_w$ with probability~$1$. 
            Similarly, let $G'(w)$ be the \emph{true} expected marginal payoff at node $w$ when $\alg'$ is executed at node~$w$; that is, $G'(u) = 0$ and $G'(w) = G(I'_w, a_w)$ for $w \neq u$. 
            Since $G$ is non-decreasing in its first parameter,
            $G'(w) \geq G(w)$ for any node $w \in T^*$; in particular, $0 = G'(u) > G(u)$ by assumption. From Equation~\eqref{equ:payoff_by_marginal}, we conclude that $\P(\alg') > \P(\alg^*)$, a contradiction to the optimality of~$\alg^*$.
        \end{proof}
    
\subsection{Omitted Proofs from Section~\ref{sec:ptas:discretization}}\label{sec:ptas:app_discret}

    \lemmaDiscret*
    
    \begin{proof}
        Fix a deterministic optimal policy $\alg$ on~$\Ins$ and let $T$ be its decision tree. 
        We construct a randomized \Sdiscretized\ policy $\Dalg$ which simulates $\alg$ when running on the \Sdiscretized\ instance $\DIns$. 
        Whenever $\Dalg$ probes a \Sdiscretized\  variable $\DX_i$, $\Dalg$ randomly samples a value $\vali$ from $F_{i \mid \Dis(\vali) =\DX_i}$,  the distribution of $\vali$ conditioning on that it discretizes to $\DX_i$. 
        $\Dalg$ then feeds $\alg$ the value $\vali$ to simulate it.
        
        It is straightforward that $\Dalg$ can be represented by the same decision tree $T$ such that $\Dalg$ reaches each node in $T$ with the same probability as $\alg$ does. 
        Let $\SG(u)$ be the marginal payoff at a node~$u$ when $\Dalg$ is executed on~$\DIns$, and let $G(u)$ be that of $\alg$ on the same node~$u$ when it is executed on~$\Ins$.  
        The state of $\alg$ and~$\Dalg$ differs by at most $\eps\threshold$ on any node, so $\SG(u) \geq G(u) - \eps^2 \threshold$ since $G$ is Lipschitz.
        By Equation~\eqref{equ:payoff_by_marginal}, we have $$\P(\Dalg, \DIns) = \sum_{u \in T}
        \SG(u) \cdot \Phi(u) \geq
        \sum_{u \in T}
        (G(u) - \eps^2 \cdot \threshold) \cdot \Phi(u) \geq 
        \P(\alg, \Ins) - O(\eps^2) \cdot \threshold \; .$$
        This proves the first statement.
        
        For the second statement, notice that for an \Sdiscretized\ policy $\Dalg$, by definition its decision tree is the same when it is executed on $\Ins$ and on~$\DIns$; in particular, each node is reached with the same probability. 
        Since $\val_i \geq \Dis(\vali)$, we have $\P(\Dalg, \DIns) \leq \P(\Dalg, \Ins)$.
    \end{proof}

    \lemmaNoTwoJackpot*
    
        \begin{proof}
                
            Since $\OPT \geq \max_i \Ex{\vali}$, after a box with value $v^* \geq \threshold$ is opened, an optimal policy never exercises the option to take a box without inspection, and the (only) optimal policy is to follow the index-based strategy on the remaining boxes with indices at least $v^*$. 

            Recall from the proof of Theorem~\ref{thm:index} that the profit of the index-based strategy is $$\Ex{\max_{i} \amorval_i} \leq \OPT,$$ where $\amorval_i = \min \{\vali, \gittins_i\}$. 
         Since $\OPT \leq \eps \threshold$, by Markov inequality, with probability at $\eps$, $\max_i \amorval_i \geq \threshold$.
            For a box~$i$ with $\gittins_i \geq v^* \geq \threshold$, $\vali \geq \threshold$ implies $\amorval_i \geq \threshold$.
            Therefore, with probability at most $\eps$, any such box has value at least~$\threshold$.
        \end{proof}

    \lemmaApproxF*
        \begin{proof}
            The first inequality is obvious. 
            For the second, let $M_i$ be the event $i = \argmax_{i \in S} \amorval_i - v$.\footnote{If there is a tie, break it lexicographically.}
            By definition,
                \begin{align*}
                    \PI(S, v) &= \Ex{(\max_{i \in S} \amorval_i - v)_+} \\
                            &= \sum_{i \in S} \Prx{M_i} \cdot \Ex{(\amorval_i - v)_+ \mid M_i} \\
                            &\geq \sum_{i \in S} \Prx{\forall j \neq i, \amorval_j \leq v \land \amorval_i > v} \cdot \Ex{(\amorval_i - v)_+ \mid M_i} \\
                            &= \sum_{i \in S} \Prx{\forall j \neq i, \amorval_j \leq v}  \Prx{\amorval_i > v} \cdot \Ex{(\amorval_i - v)_+ \mid M_i} \\
                            & \geq \sum_{i \in S} (1-\eps) \Prx{\amorval_i > v} \cdot \Ex{(\amorval_i - v)_+ \mid M_i} 
                            \\
                            &\geq \sum_{i \in S} (1-\eps) \Prx{M_i} \cdot \Ex{(\amorval_i - v)_+ \mid M_i}       \\
                            &= (1-\eps) \sum_{i \in S} \Ex{(\amorval_i - v)_+}   
                            = (1- \eps) \cdot  F(S,v)
                \end{align*}
            The second inequality follows from Lemma~\ref{lem:no2jactpot}.
        \end{proof}

    \claimLPNOIPNOI*
    
    \begin{proof}
	  Actions taken by~$\alg^L$ are valid actions in $B$: $a^0_i$ for opening box~$i$, and $a^1_i$ for taking box~$i$ without inspection.
	  Since $\UP(v) \geq v$ for any~$v$, a simple induction shows that, after taking the same sequence of actions, the state in $B^L$ is no less than the state in~$B$, the latter being the largest value revealed so far.  
	  The marginal payoff of an action in $B$ is given by the value increase on top of its state, and therefore the same action yields weakly less marginal payoff in~$B^L$.
            
        \end{proof}

    \lemmaIndexLPNOI*
        \begin{proof}
    	  For any $i \in S$, if $\Ex{G^L(v, a^0_i)} > 0$, then $\gittins_i > v \geq \threshold$.  
    	  Therefore, by Lemma~\ref{lem:no2jactpot}, each such box~$i$ is opened by the quasi-index policy with probability at least $1 - \eps$.  
    	  The marginal payoff of the policy is therefore at least $(1 - \eps) \sum_{i \in S} [\Ex{(\vali - v)_+} - c_i]_+ = (1 - \eps) \sum_{i \in S} (\amorval_i - v)_+ = (1 - \eps) F(S, v)$.
        \end{proof}

    \lemmaHighDiscret*
    
    \begin{proof}
	  The first inequality results from Claim~\ref{cl:LPNOI2PNOI}.  
	  We only need to show $\OPT^L \geq (1 - \eps) \OPT$.
	  Let $\alg^*$ be an optimal policy for~$B$, and let $T^*$ be its decision tree. 
	  Let $\SUB$ be the set of nodes in~$T^*$ where the policy first sees a revealed value at least~$\threshold$.  
	  Formally, $\SUB \coloneqq \{u \in T^*: I_u \geq \threshold, I_{\fa(u)} < \threshold\}$, where $\fa(u)$ is the father node of $u$. 
            
            Recall that $\Phi(u)$ denotes the probability of $\alg^*$ reaching a node $u$ in~$T^*$. 
            Let $S_u$ be the set of boxes not opened yet when node~$u$ is reached.
            By equation~\eqref{equ:payoff_by_marginal}, we have 
            \begin{align}
                \P(\alg) &= \sum_{u \in T^*} G(u) \cdot \Phi(u) = \sum_{\substack{u \in T^* \\ I_u < \threshold} } G(u) \cdot \Phi(u) + \sum_{u \in Q} \PI(S_u, I_u) \cdot \Phi(u) 
                \nonumber\\
                &\leq \sum_{\substack{u \in T^* \\ I_u < \threshold} } G(u) \cdot \Phi(u) + \sum_{u \in Q} F(S_u, I_u) \cdot \Phi(u) .\label{equ:opt_profit}
            \end{align}
            
	    Let $\alg^L$ be the following policy on~$B^L$: 
	    $\alg^L$ copies the behavior of $\alg$ as long as the value seen so far is at most $\threshold$; once its state reaches $\threshold$, $\alg^L$ implements a quasi-index policy.  
            
            Let $T^L$ be the decision tree of~$\alg^L$.  
	    To distinguish the notations from those in PNOI, we denote a typical node in~$T^L$ as $u^L$, write $I^L_{u^L}$ as its state, $S^L_{u^L}$ the set of boxes not opened yet, $\Phi^L(u^L)$ the probability of $\alg^L$ reaching~$u^L$, and $\P^L(\alg^L)$ the expected payoff of $\alg^L$ on~$B^L$.
	    Similarly, let $\SUB^L$ be the set of nodes in~$T^L$ with states at least $\threshold$ and whose parents have states below $\threshold$.
            
            $T^L$ is idential to $T^*$ from the root down to the nodes in~$Q^L$; each node $u \in T^*$ with $I_u \leq \threshold$ has an image $\image(u) \in T^L$, with $\alg^L$ reaching $\image(u)$ with probability $\Phi(u)$, taking the same action as $\alg^*$ does on $u$, and making expected profit.
            For a node $u \in \SUB$, let $R(u)$ be the node in~$\SUB^L$ such that when $\alg^*$ reaches~$u$, $\alg^L$ reaches $R(u)$.  
            Then $\UP(I_u) = I^L_{\image(u)}$, and $S^L_{\image(u)} = S_u$;
            for any $u^L \in Q^L$, $\Phi^L(u^L) = \sum_{u \in \image^{-1}(u^L)} \Phi(u)$.
	    Inheriting the notation from Lemma~\ref{lem:index_LPNOI}, the expected additional utility $\alg^L$ makes after it reaches a node $u^L \in \SUB^L$ is $\PI^L(S^L_{u^L}, I^L_{u^L})$. 
            By Lemma~\ref{lem:index_LPNOI} and the definition of~$\UP(\cdot)$, for any $u \in \image^{-1}(u^L)$, 
            \begin{align*}
            \PI^L(S^L_{u^L}, I^L_{u^L}) &\geq (1- \eps) \cdot F(S^L_{u^L}, I^L_{u^L}) 
            = (1 - \eps) \cdot F(S_u, \UP(I_u))
            \\
            &\geq (1 - \eps) \cdot (F(S_u, I_u) - \eps \cdot \OPT)\\
            &= F(S_u, I_u) - O(\eps) \cdot \OPT.
            \end{align*}
            
            Using inequality~\eqref{equ:opt_profit}, we have \begin{align*}
                \P(\alg) - \P^L(\alg^L) 
	      & \leq \sum_{u \in Q} \Phi(u) F(S_u, I_u) - \sum_{u^L \in \SUB^L} \Phi^L(u^L) \PI^L(S^L_{u^L}, I^L_{u^L})  \\
		&= \sum_{u \in Q} \Phi(u) \cdot \left(F(S_u, I_u) - \PI^L(S^L_{R(u)}, I^L_{R(u)}) \right) \leq O(\eps) \cdot \OPT 
            \end{align*}
            
            Therefore $\OPT^L \geq \P^L(\alg^L) \geq (1-O(\epsilon)) \cdot \OPT$.
        \end{proof}
        
    \theoremGeneralPTAS*
    \begin{proof}
            Given a PNOI instance $\Ins$, 
            denote the corresponding \Sdiscretized\  instance as $\Ins^S$. By Lemma~\ref{lem:discret}, we know $\OPT^S \geq (1-O(\eps))\cdot \OPT$, where $\OPT^S$ is the optimal expected profit of $\Ins^S$.
            
            We further discretize the values above  $\threshold$ in $B^S$ using the discretization technique from Section~\ref{sec:ptas:VH_discrete}. 
            Let $B^{SL}$ be the resulting instance.
            From Lemma~\ref{lem:high_discret}, we have $\OPT^{SL} \geq (1- O(\eps)) \cdot \OPT^S = (1-O(\eps))\cdot \OPT$, where $\OPT^{SL}$ is the optimal expected profit of $\Ins^{SL}$.
            
            
            The number of states in $\Ins^{SL}$ is $O(1 / \eps)$. 
            So we could apply Theorem~\ref{thm:jian} to  $\Ins^{SL}$ and get a policy $\alg^L$ with $\P^L(\alg^L, \Ins^{SL}) \geq (1- O(\eps)) \cdot \OPT^{SL}$.
            
            Finally, using Lemma~\ref{lem:discret} and Claim~\ref{cl:LPNOI2PNOI}, we prove that $\P(\alg_L, \Ins) \geq (1- O(\eps)) \cdot \OPT$.
        \end{proof}

\Xcomment{
    \subsection{Proof of Lemma~\ref{lem:trun_ptas23} and Lemma~\ref{lem:trun_ptas}}\label{sec:ptas:trun_appendix}
    
    \begin{proof}
        Define a truncated version of \Pandora called \TPandora, such that \TPandora will immediately end once a value greater than $\threshold$ is revealed. In other words, any valid policy for the \TPandora problem is a valid \emph{truncated} policy for the \Pandora problem and vice versa.
                
                \TPandora could be modeled as an SSDP, with the following modifications to the previous SSDP model of \Pandora in Section~\ref{sec:ptas:prob},
                \begin{enumerate}
                    \item All the states with value greater than $\threshold$ are replaced by a new state $\TState \coloneqq \threshold$.
                    \item The state transition function $f_T$ in \TPandora will be truncated accordingly, namely,
                    \begin{displaymath}
                        \begin{array}{ll}
                             f_T(I, a_i^0) &= \min \{\max\{I, \val_i\}, \TState \}; \\
                             f_T(I, a_i^1) &= \TState.
                        \end{array}
                    \end{displaymath}
                    \item The marginal payoff function $g_T$ is set to $0$ once reaching the state $\TState$, namely, 
                    \begin{displaymath}
                         g_T(\TState, \cdot) = 0 \; .  \\
                    \end{displaymath}
                        
                \end{enumerate}
                
                Therefore, we could apply Theorem~\ref{thm:const} and the discretization technique to the \TPandora model, and get a polynomial-time computable policy $\alg_T$ with expected payoff at least $\OPT_T - \eps \cdot \OPT$, where $\OPT_T$ is the optimal expected payoff of the \TPandora problem. 
                
                Noticing that $\P(\alg^*_T) \leq \OPT_T$ and $\OPT_T \leq \OPT$, we conclude that 
                $$ \P(\alg_T) \geq  \OPT_T - \eps \cdot \OPT \geq \P(\alg^*_T) - \eps \cdot \OPT \; .$$
    \end{proof}
}

\end{document}